\RequirePackage{ifpdf}
\newif\ifpdf
\ifx\pdfoutput\undefined
\pdffalse
\else
\pdfoutput=1
\pdftrue
\fi

\ifpdf
\documentclass[pdftex]{llncs}
\else
\documentclass[dvips]{llncs}
\fi

\usepackage{fullpage}
\usepackage{fancyhdr}
\usepackage[T1]{fontenc}
\usepackage{amsmath}
\usepackage{amssymb}
\usepackage{graphicx}
\usepackage{placeins}
\usepackage{color}
\usepackage{colortbl}
\usepackage[dvips]{epsfig}
\usepackage{xspace}
\usepackage{algorithm}
\usepackage{algorithmicx}
\usepackage{algpseudocode}
\usepackage{float}
\newfloat{algorithm}{t}{lop}
\usepackage{floatflt}
\definecolor{grey}{gray}{0.7}
\usepackage{url}
\usepackage{lineno}
\usepackage{wrapfig}

\usepackage{ae, aecompl}
\usepackage{colortbl}
\usepackage{numprint}

\newcommand{\eg}{e.\,g.\xspace}

\newcommand{\ie}{i.\,e.\xspace}

\newcommand{\degree}{\operatorname{deg}}

\newcommand{\bigO}{\mathcal{O}}

\newcommand{\etal}{\textit{et al.}\xspace}

\newcommand{\wrt}{w.\,r.\,t.\xspace}

\newcommand{\metis}{\textsc{METIS}\xspace}
\newcommand{\gpmetis}{\textsc{gpMetis}\xspace}
\newcommand{\ndmetis}{\textsc{ndMetis}\xspace}

\newcommand{\kahip}{\textsc{KaHIP}\xspace}
\newcommand{\initial}{\textsc{Initial}\xspace}
\newcommand{\initialM}{\textsc{Initial}_{\textsc{g}} \xspace}
\newcommand{\initialMO}{\textsc{Initial}_{\textsc{n}} \xspace}
\newcommand{\random}{\textsc{Random}\xspace}
\newcommand{\randomM}{\textsc{Random}_{\textsc{g}} \xspace}
\newcommand{\randomMO}{\textsc{Random}_{\textsc{n}} \xspace}
\newcommand{\rcm}{\textsc{RCM}\xspace}
\newcommand{\rcmM}{\textsc{RCM}_{\textsc{g}} \xspace}
\newcommand{\rcmMO}{\textsc{RCM}_{\textsc{n}} \xspace}
\newcommand{\durebi}{\textsc{DRB}\xspace}
\newcommand{\durebiM}{\textsc{DRB}_{\textsc{g}} \xspace}
\newcommand{\durebiMO}{\textsc{DRB}_{\textsc{n}} \xspace}
\newcommand{\greedyall}{\textsc{GreedyAll}\xspace}

\newcommand{\greedyallc}{\textsc{GreedyAllC}\xspace}
\newcommand{\greedyallcM}{\textsc{GreedyAllC}_{\textsc{g}} \xspace}
\newcommand{\greedyallcMO}{\textsc{greedyAllC}_{\textsc{n}} \xspace}
\newcommand{\greedymin}{\textsc{GreedyMin}\xspace}
\newcommand{\greedyminM}{\textsc{GreedyMin}_{\textsc{g}} \xspace}
\newcommand{\greedyminMO}{\textsc{GreedyMin}_{\textsc{n}} \xspace}
\newcommand{\greedyminc}{\textsc{GreedyMinC}\xspace}
\newcommand{\greedymincM}{\textsc{GreedyMinC}_{\textsc{g}} \xspace}
\newcommand{\greedymincMO}{\textsc{GreedyMinC}_{\textsc{n}} \xspace}
\newcommand{\walshawlarge}{\textsc{WalshawLarge}\xspace}
\newcommand{\complexnets}{\textsc{ComplexNets}\xspace}

\newcommand{\scotch}{\textsc{Scotch}\xspace}

\newcommand{\NN}{\mbox{\rm I$\!$N}}

\newcommand{\intmax}{\textsc{int\_max}\xspace}

\newcommand{\anoe}[1]{\textcolor{blue}{[AN: #1]}\xspace}
\newcommand{\hmey}[1]{\textcolor{red}{[HM: #1]}\xspace}
\newcommand{\rgla}[1]{\textcolor{grey}{[RG: #1]}\xspace}
\renewcommand{\anoe}[1]{}
\renewcommand{\hmey}[1]{}
\renewcommand{\rgla}[1]{}

\renewcommand{\baselinestretch}{1.03}

\begin{document}

%
%
\title{Algorithms for Mapping Parallel Processes onto\\Grid and Torus Architectures}

\author{Roland Glantz \and Henning Meyerhenke \and Alexander Noe}

\institute{
Karlsruhe Institute of Technology (KIT), Karlsruhe, Germany}
\maketitle
%
\vspace{-4mm}

%
%
\begin{abstract}
Static mapping is the assignment of parallel processes to the
processing elements (PEs) of a parallel system, where the assignment
does not change during the application's lifetime. In our scenario we
model an application's computations and their dependencies by an
application graph. This graph is first partitioned into (nearly)
equally sized blocks. These blocks need to communicate at block
boundaries.  To assign the processes to PEs, our goal is to compute a
communication-efficient bijective mapping between the blocks and the
PEs.

This approach of partitioning followed by bijective mapping has many
degrees of freedom. Thus, users and developers of parallel
applications need to know more about which choices work for which
application graphs and which parallel architectures. To this end, we
not only develop new mapping algorithms (derived from known greedy
methods).  We also perform extensive experiments involving different
classes of application graphs (meshes and complex networks),
architectures of parallel computers (grids and tori), as well as
different partitioners and mapping algorithms. Surprisingly, the
quality of the partitions, unless very poor, has little influence on
the quality of the mapping.

More importantly, one of our new mapping algorithms always yields the
best results in terms of the quality measure maximum congestion when
the application graphs are complex networks.  In case of meshes as
application graphs, this mapping algorithm always leads in terms of
maximum congestion \emph{and} maximum dilation, another common quality
measure.
\end{abstract}


%
%
\section{Introduction}
\label{sec:intro}
Symmetric dependencies of computations within a parallel application can be modeled by an
undirected graph $G_a$, called \emph{application graph}, \eg the
mesh of a numerical simulation. Iterative algorithms in such a
simulation act upon the vertices of $G_a$, and for each such vertex
require the values of the neighboring vertices from the previous
iteration. Thus, a vertex of $G_a$ represents some computation, and an
edge of $G_a$ indicates a dependency between computations, \ie
an exchange of data. It is important to note that this modeling is not restricted 
to simulations at all. In fact, the nodes of $G_a$ could represent arbitrary parallel processes
and the edges symmetric communication requirements between the processes.

Typically, running an application on computers with distributed parallelism requires the application graph to be spread over 
the computer's processing elements. One
way to carry out this task, called \emph{static mapping}, is to (i)
partition the application graph $G_{a}$ into blocks of equal size (or of equal weight
in case the computational requirements at the nodes are not homogeneous) for
load balancing purposes and (ii) map the blocks of $G_{a}$ onto the
processing elements (PEs) of a parallel computer, see
Figure~\ref{fig:overview}. 
Mapping may involve the communication
graph $G_{c}$, whose vertices represent the blocks of $G_{a}$'s
partition and whose edges indicate block neighborhood and therefore
communication between different PEs. 

The parallel computer is often represented as a graph $G_p$, called \emph{processor
  graph} (or topology graph), the vertices of which represent the PEs,
and the edges of which represent physical communication links between
the PEs. We require that $G_c$ has the same number of vertices as
$G_p$ and make the assumption that $G_p$ is sparse,
which is true for many real architectures today~\cite{top500-13-06}.
In this paper we address the problem of finding a bijective mapping
$\Pi$ of $G_c$'s vertex set onto $G_p$'s vertex set
(processors) that is communication-efficient. We refer to $\Pi$ as
\emph{bijective topology mapping} or simply \emph{mapping}. One can also
see the problem as embedding the guest graph $G_c$ into the host graph $G_p$.

\begin{figure}[tb]
\centering{}(a) \includegraphics[width=2.3cm]{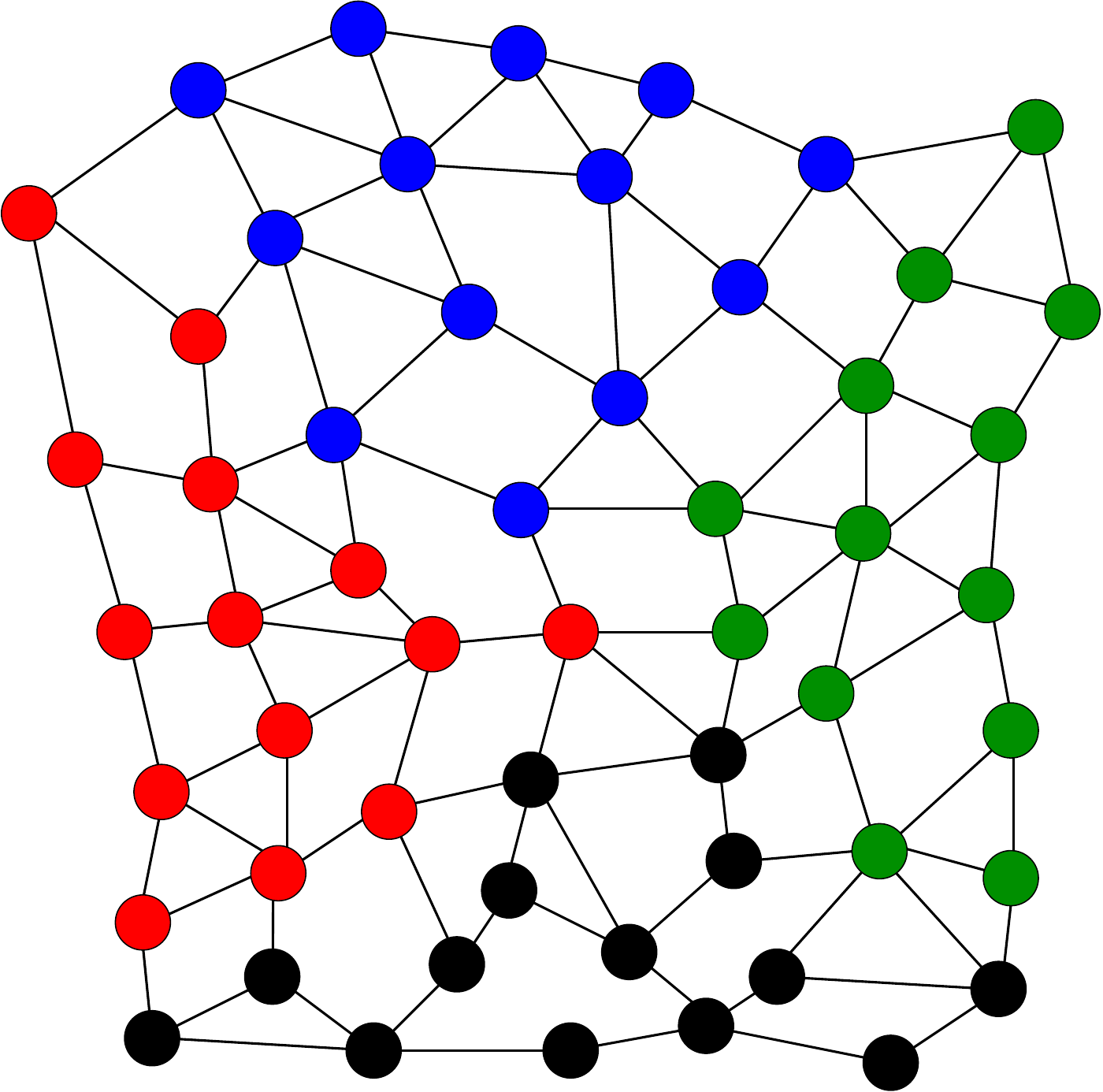} (b)
\includegraphics[width=2.3cm]{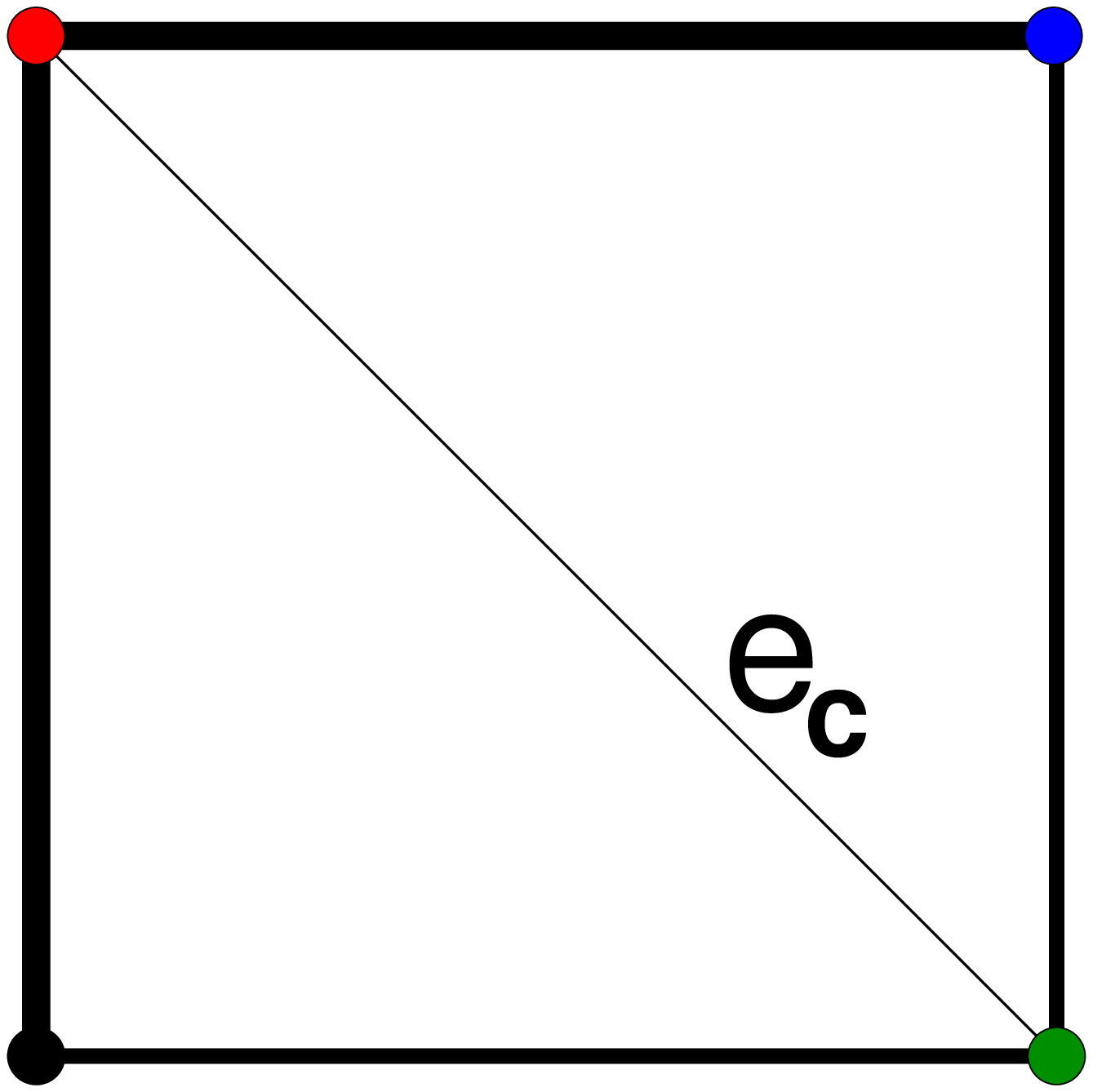} (c)
\includegraphics[width=2.3cm]{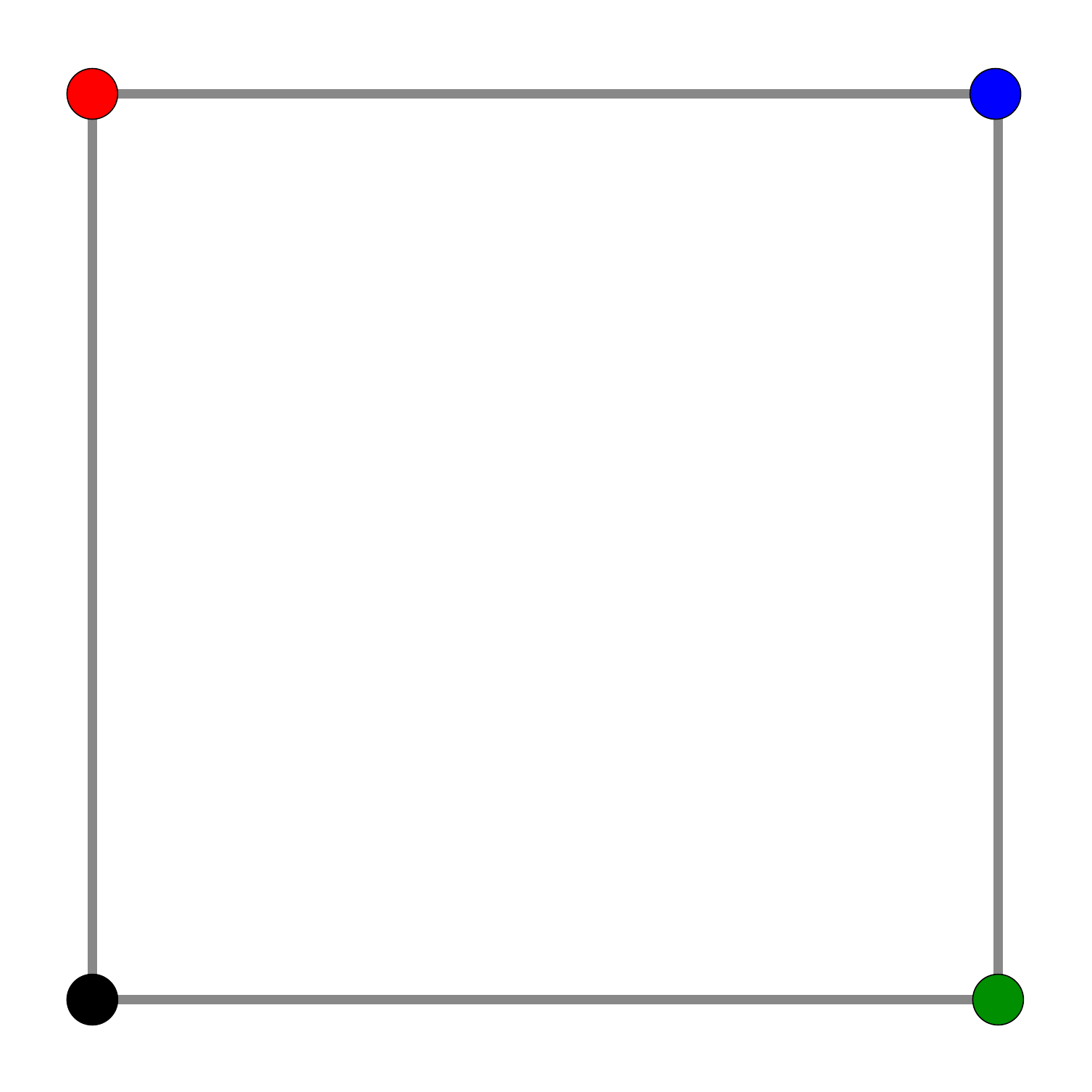}\caption{\label{fig:overview}(a)
  Application graph $G_{a}$ with $4$-way partition indicated by
  colors. (b) Com\-mu\-ni\-ca\-tion graph $G_{c}$ induced by $G_{a}$
  and the partition. $G_c$ expresses the neighborhood relations of
  $G_{a}$'s blocks. Edge weights (shown through width) indicate
  communication volumes between blocks. (c) Processor graph $G_{p}$.
  Nodes and edges represent the PEs and the communication links,
  respectively. Com\-mu\-ni\-ca\-tion between the green and the red
  block in $G_{c}$, \ie via $e_{c}$, requires two hops in $G_{p}$.  }
\vspace{-3ex}
\end{figure}


\noindent \emph{Motivation.} Communication costs are crucial for the
scalability of many parallel applications. Static mapping, in turn, is
crucial when it comes to keeping communication costs under control
through (i) providing a partitioning with few edges between blocks and
(ii) mapping nearby blocks onto nearby PEs: due to the sparse nature
of many large-scale parallel computers,
communication costs may vary by several orders of magnitude depending
on the distance between the PEs
involved~\cite{Teresco2000269}. Also, numerous recent applications involve
massive \emph{complex networks} such as social networks or web graphs~\cite{costa2011analyzing}. These
networks usually lead to denser communication graphs and make improved
mapping strategies even more desirable.

\noindent \emph{Contribution.}
We investigate numerous algorithms for static mapping,
the scenario being that an application graph is first partitioned into
blocks, followed by a bijective mapping of the blocks onto the nodes
of a processor graph. The graph partitioners we employ are the
state-of-the-art packages \metis~\cite{Karypis13a} and
\kahip~\cite{Sanders2013a}. While \metis is widely used for graph
partitioning and has been employed for mapping before, it is the first
time that the high-quality partitioner \kahip is used in the mapping
context.

To assess and improve the performance of mapping algorithms, we
implement several state-of-the-art methods. Moreover and more
importantly, we develop and implement two new algorithms as
straightforward, yet very effective adaptations of existing greedy
algorithms.

The three most striking results of our extensive mapping experiments
on meshes and complex networks as application graphs, as well as grids
and tori as processor graphs, are: First, the strengths
and weaknesses of the mapping algorithms are, to a large extent,
independent of the class of application graphs (mesh or complex
network) and the processor graphs. Second, the graph partitioner and
its partitioning quality is of minor importance for the quality
of the mapping. Third, for complex networks as application graphs,
one of our new mapping algorithms always yields the best quality in terms 
maximum congestion. In case of meshes, this mapping algorithm always 
leads in terms of maximum congestion \emph{and} maximum dilation.

%
%
\section{Preliminaries}
\label{sec:prelim}
%
\subsection{Problem Description}
\label{sub:problem}

We represent the communication of a parallel appli\-ca\-tion as a
graph $G_c = (V_c, E_c, \omega_c)$, where a weight $\omega_c(\{u,
v\})$, $\{u, v\} \in E_c$, indicates the volume of communication
between $u$ and $v$, \ie between the corresponding blocks of the
application graph.

The parallel computer takes the form of a graph $G_p = (V_p, E_p,
\omega_p)$, the \emph{processor graph}. Here, $\omega_p~:~E_p \mapsto
\NN$ indicates the bandwidths of the physical communication links. We
require $\vert V_p \vert = \vert V_c \vert$.

Our aim is to find a bijective topology mapping (short \emph{mapping})
$\Pi: V_c \mapsto V_p$ that minimizes the overhead due to
communication between the processes. A first graph-theoretic
definition of the overhead (costs) was given
in~\cite{Rosenberg1980a}. In the following we present three aspects of
overhead (for more in-depth definitions see~\cite{hoefler-topomap}).

An edge $e_c = \{u_c, v_c\}$ of $G_c$ gives rise to communication
between $\Pi(u_c)$ and $\Pi(v_c)$ on $G_p$. Sending a unit of
information along a path $P$ in $G_p$ with edges $e_1, \dots e_l$
takes time at least $t(P) = \sum_{i=1}^l(1 / \omega_p(e_i))$. Sending
all information via an edge $e_c = \{u_c, v_c\} \in E_c$, \ie from
processes in $u_c$ to processes in $v_c$, then takes time at least

\begin{align}
\label{eq:dt}
d(e_c) &= d(e_c, \Pi) = \omega_c(u_c, v_c)~t(\Pi(u_c),
\Pi(v_c)),\mbox{ where}\nonumber\\
t(u_p, v_p) &= \min(t(P)~\vert~\mbox{$P$ connects
  $u_p$ and $v_p$})
\end{align}
Thus, \emph{maximum} and \emph{average dilation}, defined as
\begin{align}
\label{eq:dil} 
mD(\Pi) &= \max_{e_c \in E_c} d(e_c) \mbox{~and}\\
aD(\Pi) &= (\sum_{e_c \in E_c} d(e_c)) / \vert E_c \vert \mbox{~,}
\end{align} 
respectively, provide lower bounds for the communication time of a
parallel application, $mD(\Pi)$ being the tighter lower bound.

When multiple messages are exchanged at the same time, more than one
of them may be routed via the same edge. Hence, if $c(e)$ denotes
the total volume of communication routed via $e \in E_p$, divided by
the bandwidth $\omega(e)$, then the maximum (weighted) congestion
\begin{equation}
\label{eq:maxCon}
mC(\Pi) := max_{e_c \in E_c} c(e_c)
\end{equation} 
provides another lower bound for the time. Minimizing $mD(\Pi)$,
$aD(\Pi)$ and $mC(\Pi)$ is NP-hard, cf.\ Garey and
Johnson~\cite{Garey:1979:CIG:578533} and more recent
work~\cite{hoefler-topomap,ManKim1991246}. Due to the problem's
complexity, exact mapping methods are only practical in special
cases. Leighton's book~\cite{Leighton92introduction} discusses
embeddings between arrays, trees, and hypercubes.

As in previous studies~\cite{hoefler-topomap}, we assume that the
routing algorithm sends the messages on uniformly distributed shortest
paths in $G_p$. In particular, the routing algorithm is oblivious to
the utilization of the parallel system. 


\subsection{Graph partitioning}
\label{sub:part_quality}

Given a graph $G=(V,E)$ and a number of blocks $k > 0$, the 
graph partitioning problem asks for a division of $V$ into $k$ pairwise
disjoint subsets $V_1, \dots, V_k$ (\emph{blocks}) such that no block
is larger than
%
%
$(1+\varepsilon)\cdot \left\lceil\frac{|V|}{k}\right\rceil,$
%
where $\varepsilon \geq 0$ is the allowed imbalance. The most widely
used objective function 
is the \emph{edge cut} (whose minimization is
$\mathcal{NP}$-hard~\cite{Garey:1979:CIG:578533}), \ie, the total weight of the edges between
different blocks.  Yet, a more important factor for modeling the communication cost of parallel 
iterative graph algorithms seems to be the \emph{maximum communication volume} (MCV)~\cite{Hendrickson_graphpartitioning}, 
which has received growing attention recently, \eg in the 10th DIMACS Implementation Challenge
on graph partitioning. 
MCV considers the worst
communication volume taken over all blocks $V_p$ ($1 \leq p \leq k$)
and thus penalizes imbalanced communication:
$MCV(V_1, \dots, V_k) := \max_p \sum_{v \in V_p} |\{ V_i ~|~ \exists \{u, v\} \in E \mbox{ with } u \in V_i  \neq V_p\}|.$

\section{Related Work}
\label{sub:related}
In this section we give a brief overview of algorithms for static
mapping. 
More on topology mapping can be found
in~\cite{Aubanel09resource,6495451} and particularly in Pellegrini's
survey~\cite{Pellegrini11static}.

It should be mentioned that partitioning and mapping can be done
simultaneously, \ie communication between PEs is taken into account
already during
partitioning~\cite{DBLP:journals/fgcs/WalshawC01,HuangAB06pagrid,MoulitsasK08architecture}.
In this paper, however, we focus on the complementary approach where
partitioning and topology mapping form different stages of a software
pipeline.

One can apply a wide range of optimization techniques to the topology
mapping problem. Hoefler and Snir~\cite{hoefler-topomap} employ (among
others) the Reverse Cuthill-McKee (RCM) algorithm, originally devised
for minimizing the bandwidth of a sparse matrix~\cite{Cuthill69a}. If
both $G_c$ and $G_p$ are sparse, the simultaneous optimization of both
graph layouts can lead to good mapping
results~\cite{Pellegrini07scotch}.

A common approach to static mapping, \ie, partitioning and topology
mapping combined, is to recursively partition $G_a$ \emph{and} $G_p$
in the same fashion, \ie such that the number of blocks and sub-blocks
per block is equal on each level~\cite{Pellegrini94static}. Such a
hierarchical approach to mapping may take into account the actual
hierarchy of a heterogeneous multi-core
cluster~\cite{chan2012impact}. Typically, the number of sub-blocks per
block is small. Thus, on the scope of an individual block, an optimal
mapping of a block's sub-blocks can be found by evaluating all
possibilities. If the number of sub-blocks is two, the method is
called \emph{dual recursive bisection}. It has been shown effective in
the software \scotch\cite{Pellegrini07scotch}. While an optimal
mapping of a block's sub-blocks on the scope of an individual block is
not an issue in dual recursive bisection, neighboring relations
between sub-blocks of different blocks still pose a challenge.
In this paper we apply dual recursive bisection to the pair $(G_c,
G_p)$ instead of $(G_a, G_p)$. This (basic) form of dual recursive
bisection does not take into account neighboring relations between the
sub-blocks of different blocks (as in~\cite{hoefler-topomap}).

Greedy approaches such as the ones by Hoefler and Snir~\cite{hoefler-topomap}
and Brandfass~\etal~\cite{Brandfass2013372} build on the idea of
increasing a mapping by successively adding new maps $v_c \rightarrow
v_p$ such that (i) $v_c$ has maximal communication volume with one or
all of the already mapped vertices of $G_c$ and (ii) $v_p$ has minimal
distance to one or all of the already mapped vertices of $G_p$. For
more details see Sections~\ref{subsec:greedy}.

Hoefler and Snir~\cite{hoefler-topomap} compare RCM, DRB and a greedy
approach experimentally on abstractions of three real architectures. While their
results do not show a clear winner, they confirm previous
studies~\cite{Pellegrini11static} in that performing mapping at all is
worthwhile.  It is important to note, however, that Hoefler and Snir
perform mapping from reordered matrices, not from partitioned graphs
as we do here.

Many metaheuristics have been used to solve the mapping problem.
U\c{c}ar \etal~\cite{Ucar200632} implement a large variety of methods
within a clustering approach, among them genetic algorithms, simulated
annealing, tabu search, and particle swarm optimization. The authors
require, however, that the processor graph is homogeneous, \ie $t(u_p,
v_p)$ depends only on whether $u_p = v_p$ or not. Our approach is more
general than theirs in that we allow $t(u_p, v_p)$ to take different
values for $u_p \not= v_p$ (see Equation~\ref{eq:dt}). 

Bhatele \etal~\cite{Bhatele:2011:AHT:2063384.2063486} discuss topology-aware
mappings of different MPI communication patterns on
emerging architectures. Better mappings avoid communication hot spots
and reduce communication times significantly. Geometric information
can also be helpful for finding good mappings on regular architectures
such as tori~\cite{6063073}.

%
%
\section{Methods for Topology Mapping}
\label{sec:algo}
The simplest topology mapping is the identity, \ie when block $i$ of
the application graph (or vertex $i$ of the communication graph $G_c$)
is mapped onto node $i$ of the processor graph $G_p$, $1 \leq i \leq
k$. We refer to this mapping as $\initial$. It depends on how the
graph partitioner, in our case \metis or \kahip, numbers the blocks
and on how the nodes of $G_p$ are numbered. In our experiments $G_p$
is a 2D or 3D grid or torus since such topologies are used in real architectures,
\eg tori for BlueGene~\cite{BlueGene02overview}.
The nodes are ordered
lexicographically \wrt the nodes' canonical integer coordinates. We
also carry along a mapping called $\random$, where the bijection $\Pi:
\{1, \dots, k\} \mapsto \{1, \dots, k\}$ is random. The latter is done
for comparison purposes, keeping in mind that $\random$ is usually a
very bad solution.

Four algorithms in our collection, \ie, $\rcm$, $\durebi$,
$\greedyall$ and $\greedymin$ are from the literature (for \rcm and
$\durebi$ see Section~\ref{sec:intro}
and~\cite{Cuthill69a,hoefler-topomap}). Algorithms $\greedyall$
and $\greedymin$ are described in Section~\ref{subsec:greedy} 
(also see the references therein). There we also specify the last two
algorithms, $\greedyallc$ and $\greedyminc$, which are variants of
$\greedyall$ and $\greedymin$ and which, to our knowledge, are new.

\subsection{Greedy Algorithms}
\label{subsec:greedy}
As a prerequisite for the algorithms described in this section we need
to compute $t(\cdot, \cdot)$ once for a given processor graph $G_p$
(see Equation~\ref{eq:dt}). Using Johnson's
algorithm~\cite{Johnson77a,Cormen2001a} we can do so in time
$\bigO(\vert V_p \vert^2 \log \vert V_p \vert$ $+$ $\vert V_p \vert
\vert E_p \vert)$. Since $G_p$ is sparse, this amounts to $\bigO(\vert
V_p \vert^2 \log \vert V_p \vert) = \bigO(\vert V_c \vert^2 \log \vert
V_c \vert)$. This running time is not included in the running times
for the greedy algorithms in this section, as $t(\cdot, \cdot)$ is
computed only once for a given processor graph.

The mapping algorithm $\greedyall$ consists of the ``construction
method'' proposed in~\cite{Brandfass2013372}. Using our terminology,
the algorithm starts by picking a node $v^0_c$ of $G_c$ such that
$\sum_{e = \{v^0_c, v_c\} \in E_c} \omega(e)$ is maximal, \ie $v^0_c$
is a vertex whose communication with neighboring vertices is
heaviest. Then, it computes for each vertex $v_p$ of $G_p$ the term
$\sum_{u_p \in V_p} t(u_p, v_p)$. Here, $t(u_p, v_p)$ is the (minimum)
time needed to send a unit of information from $u_p$ to $v_p$ (see
Section~\ref{sub:problem}). A vertex $v^0_p$ for which this sum is
minimal (a most central node in $G_p$ \wrt communication time) then
becomes the vertex onto which $v^0_c$ is mapped. The experiments of
this paper involve processor graphs which are grids and tori. On the
latter all nodes are equally central.

The remaining pairs $(v^i_c, v^i_p)$, $i \geq 1$, are formed as
follows. First, a not yet mapped vertex $v^i_c$ of $G_c$ is found such
that $\sum^{i-1}_{j=0}\omega_c(\{v^j_c, v^{i}_c\})$ is maximal, \ie
$v^i_c$ is a vertex that communicates most heavily with the already
mapped vertices. Then, a not yet mapped vertex $v^i_p$ of $G_p$ is
found such that $\sum^{i-1}_{j=0}t(v^j_p, v^i_p)$ is minimal, \ie a
vertex that is most central \wrt the already mapped vertices of
$G_p$. Note that the choices of $v^i_c$ and $v^i_p$ are independent of
each other. Our implementation of $\greedyall$ has running time
$\bigO(\vert V_c \vert^2)$. This running time is achieved by updating
vectors $sum_c$ ($sum_p$) that, for each vertex $v_c$ ($v_p$) which
has not been mapped yet, stores the sum of the edge weights
(distances) to the vertices in $G_c$ ($G_p$) that have been mapped
already. We use the same two vectors in \greedyallc, see
Algorithm~\ref{algo:greedyallc}.

\label{sec:greedymin}
The mapping algorithm $\greedymin$ stems
from~\cite{hoefler-topomap}. Its general idea is the same as that
behind $\greedyall$. The only differences are that (i) $v^0_p$ is
picked randomly, (ii) $v^i_c$ ($i \geq 1$) is chosen such that
$max^{i-1}_{j=0}\omega_c(\{v^i_c, v^j_c\})$ is maximal, and (iii)
$v^i_p$ ($i \geq 1$) is chosen such that $t(v^{i-1}_p, v^i_p)$ is
minimal. Again, as in $\greedyall$, the choices of $v^i_c$ and $v^i_p$
are independent of each other. Our implementation of $\greedymin$
(which is less generic than that in~\cite{hoefler-topomap}) has
running time $\bigO(\vert V_c \vert^2)$.

\subsection{$\greedyallc$ and $\greedyminc$}
\label{sec:greedy+}
Neither $\greedyall$ nor $\greedymin$ link the choices of $v^c_i$ and
$v^p_i$. Both algorithms aim at (i) a high communication volume of
$v^c_i$ with all or one of the already mapped vertices of $G_c$ and
(ii) a high centrality of $v^p_i$ \wrt all or one of the already
mapped vertices of $G_p$. The actual increase of communication times
caused by the new pair $(v^c_i, v^p_i)$ (increase \wrt the partial
mapping defined so far) is not considered.
 
We therefore propose new variants $\greedyallc$ and
$\greedyminc$. They take this increase of communication time into
account. Specifically, the choice of $v^i_p$ depends on the choice of
$v^i_c$ (same as in $\greedyall\slash\greedymin$). Let
$v_p$ be a candidate for being mapped onto by $v^i_c$. Then, (minimal)
times of communication between $v^i_c$ and the vertices of $G_c$ that
have been mapped before, \ie $v^0_c, \dots v^{i-1}_c$, amount to

\begin{equation}
\label{eq+}
\sum(\omega_c(\{v^i_c, v_c\}~t(v_p, \Pi(v_c)))~\mid~\{v^i_c, v_c\}
\in E_c~,v_c \in \{v^0_c, \dots v^{i-1}_c\})
\end{equation}

Analogous to $\greedyallc$ and $\greedyminc$, we set $v^i_p$ to some
$v_p$ such that the expression in Equation~\ref{eq+} is
\emph{minimal}. Thus, our objective function for choosing $v_p$, \ie
Equation~\ref{eq+}, is about actual communication times and not just
distances on $G_p$. We have experimented with replacing the sum in
Equation~\ref{eq+} by the maximum and found out that this tends to
decrease the quality of the mappings. For the pseudocode of
$\greedyallc$ see Algorithm~\ref{algo:greedyallc}.

\hmey{Latex warning in pseudocode}

\begin{algorithm}[!h]
\caption{The algorithm $\greedyallc$. \newline \underline{Input}:
  Communication graph $G_c=(V_c, E_c, \omega_c)$ and processor graph
  $G_p=(V_p, E_p, \omega_p)$ with $\vert V_c \vert = \vert V_p
  \vert$.\newline \underline{Output}: Pairs $(v^{i}_c, v^{i}_p)$, $0 \leq i <
  \vert V_c \vert$, such that $\Pi~:~V_c \mapsto V_p$ defined by
  $\Pi(v^{i}_c) = v^{i}_p$ is a bijective mapping with low values of
  $mC(\Pi)$, $mD(\Pi)$ and $aD(\Pi)$.}

\label{algo:greedyallc}
\begin{algorithmic}[1]
\State Find $v^0_c \in V_c$ with maximal $\sum_{e = \{v^0_c, v_c\} \in E_c} \omega(e)$
\State Find $v^0_p \in V_p$ with minimal $\sum_{u_p \in V_p} t(u_p, v_p)$
\State Create vectors $sum_c$ and $sum_p$ of length $\vert V_c \vert$
\State Initialize entries of $sum_c$ to zero and entries of  $sum_p$ to one
\For{$i \gets 0, \dots, \vert V_c \vert - 2$}
\State $sum_c(v^i_c) \gets -1$ /* Mark $v^i_c$ as \emph{assigned} */
\State $sum_p(v^i_p) \gets \intmax$  /* Mark $v^i_p$ as \emph{assigned} */
\ForAll{$e_c=\{v^i_c, w\} \in E_c$}
\If{$sum_c[w] \geq 0$} /* $w$ is not yet assigned */
\State $sum_c[w] \gets sum_c[w] + \omega_c(e_c)$
\EndIf
\EndFor
\State Pick $v^{i+1}_c$ such that $sum_c(v^{i+1}_c)$ is maximal
\For{$j \gets 1, \dots, \vert V_p \vert - 1$}
\If{$sum_p[j] < \intmax$}
\State /* $j$ is not yet assigned */
\State $sum_p[j] \gets 0$
\ForAll{$e_c=\{v^{i+1}_c, w\} \in E_c$}
\If{$sum_c[w] < 0$}
\State /* $w$ has already been assigned, \ie $\Pi(w)$ is defined */
\State $sum_p[j] \gets sum_p[j] + \omega_c(e_c) * t(j, \Pi(w))$
\EndIf
\EndFor
\EndIf
\EndFor
\State Pick $v^{i+1}_p$ such that $sum_p(v^{i+1}_p)$ is maximal
\EndFor
\end{algorithmic}
\end{algorithm}

\begin{proposition}
 The running time of $\greedyallc$ is $\bigO(\vert V_c \vert \vert E_c \vert)$.
\label{prop:time}
\end{proposition}

\begin{proof}
The outermost loop from line 5 to line 27 and the inner loop from line
8 to line 12 take amortized time $\bigO(\vert E_c \vert)$. So does the
outer loop from line 14 to line 25 and the inner loop from
line 18 to line 23. Since the latter two loops are contained in the
outermost loop from line 5 to line 27, the running time of
Algorithm~\ref{algo:greedyallc} is indeed $\bigO(\vert V_c \vert \vert
E_c \vert)$. Even a trivial implementation of lines 13 and 26 (with
running time $\bigO(\vert V_c \vert)$) does not change the result.
\end{proof}

The running time for $\greedymin$ is the same as for $\greedyallc$
because the two algorithms differ only at lines 1 to 4, and the
running times of both algorithms are not determined by this part.
%

%
%
\section{Experiments}
\label{sec:exp}
In this section we specify our test instances, our experimental setup
and the way we evaluate the mapping algorithms.

\paragraph{Test Instances.}
\label{sub:exp-instances}
The application graphs fall into two classes: 
The class \walshawlarge consists of the eight largest graphs in
Walshaw's graph partitioning archive~\cite{SoperWC04combined}, and the
class \complexnets consists of 12 complex networks (see
Tables~\ref{tab:walshaw} and~\ref{tab:complex}). The latter form a
subset of the 15 complex networks used in~\cite{Safro2012a}
%
%
for partitioning experiments. It turned out, however, that $\kahip$
[$\gpmetis$ with $k$-way partitioning, respectively], while respecting
the allowed imbalance, occasionally generated empty blocks for the
complex network \emph{p2p-Gnutella} [\emph{as-22july06} and
  \emph{loc-gowalla\_edges}]. Using $\gpmetis$ with \emph{recursive
  bisection} instead of k-way partitioning was not an option because
$\gpmetis$ then quite often violated the balance constraint and
produced blocks heavier than $(1+\epsilon)$ times the average block
size (only on complex networks).
For each of the classes \walshawlarge and
\complexnets the benchmarking comprises the following processor graphs.
\begin{itemize}
    \item {\small 2DGrid($16 \times 16$), 2DGrid($32 \times 32$), 3DGrid($8 \times 8 \times 8$)}
    \item {\small 2DTorus($16 \times 16$), 2DTorus($32 \times 32$), 3DTorus($8 \times 8 \times 8$)}
\end{itemize}

\begin{table}[]
\caption{Meshes used for benchmarking}
\begin{center}
  \begin{tabular}{ l | r | r }
    Name & \#vertices & \#edges\\ \hline \hline
fe\_tooth  & \numprint{78136}   & \numprint{452591} \\\hline
fe\_rotor  & \numprint{99617}   & \numprint{662431} \\\hline
598a       & \numprint{110971}   & \numprint{741934} \\\hline
fe\_ocean  & \numprint{143437}  & \numprint{409593} \\\hline
144        & \numprint{144649}   & \numprint{1074391} \\\hline
wave       & \numprint{156317}   & \numprint{1059331} \\\hline
m14b       & \numprint{214765}   & \numprint{1679018} \\\hline
auto       & \numprint{448695}   & \numprint{3314611} \\\hline
  \end{tabular}
\end{center}
\label{tab:walshaw}
\end{table}

\begin{table*}[]
\caption{Complex networks used for benchmarking.}
\begin{center}
\scalebox{0.8}{
  \begin{tabular}{ l | r | r | c }
    Name & \#vertices & \#edges & Type\\ \hline \hline
PGPgiantcompo         & \numprint{10680}  & \numprint{24316}    & largest connected component in network of PGP users\\\hline
email-EuAll           & \numprint{16805}  & \numprint{60260}    & network of connections via email\\\hline
soc-Slashdot0902      & \numprint{28550}  & \numprint{379445}   & news network\\\hline
loc-brightkite\_edges & \numprint{56739}  & \numprint{212945}   & location-based friendship network\\\hline
coAuthorsCiteseer     & \numprint{227320} & \numprint{814134}   & citation network\\\hline
wiki-Talk             & \numprint{232314} & \numprint{1458806}  & network of user interactions through edits\\\hline
citationCiteseer      & \numprint{268495} & \numprint{1156647}  & citation network\\\hline
coAuthorsDBLP         & \numprint{299067} & \numprint{977676}   & citation network\\\hline
web-Google            & \numprint{356648} & \numprint{2093324}  & hyperlink network of web pages\\\hline
coPapersCiteseer      & \numprint{434102} & \numprint{16036720} & citation network\\\hline
coPapersDBLP          & \numprint{540486} & \numprint{15245729} & citation network\\\hline
as-skitter            & \numprint{554930} & \numprint{5797663}  & network of internet service providers\\\hline
  \end{tabular}}
\end{center}
\label{tab:complex}
\end{table*}

\paragraph{Experimental Setup.}
\label{sub:exp-setup}
All computations are sequential and done on a workstation with two
4-core Intel(R) Core(TM) i7-2600K processors at 3.40GHz. Our code is written
in C++ and compiled with GCC 4.7.1. 

\paragraph{Evaluation.}
\label{sub:exp-evaluation}
The benchmarking of the mapping algorithms described in
Section~\ref{sec:algo} is done separately on the classes \walshawlarge
and \complexnets. First, graphs from both classes are partitioned into
256, 512 and 1024 parts using the graph partitioner \kahip v. 0.62
({\tt http://algo2.iti.kit.edu/documents/kahip/})
\cite{dissSchulz}.
In particular,
the meshes and social networks are partitioned with the configuration
\emph{eco} and \emph{ecosocial}, respectively. The allowed imbalance
is always $1.03$, \ie $\epsilon = 3\%$. To recursively bipartition $G_c$ and $G_p$ during
$\durebi$, we also use \kahip (configurations \emph{fast} and
\emph{ecofast}, perfect balance).

Since the partitioning process depends on random choices, we run \kahip with 20
different seeds.  For each seed we construct a communication graph
$G_c$ from the partition, map $G_c$ onto all processor graphs with the
same number of vertices and then compute the minimum, the arithmetic
mean and the maximum of the mapping's runtime $t$, $mC$ (see
Equation~\ref{eq:maxCon}), $mD$ and $aD$ (see
Equation~\ref{eq:dil}). Thus we arrive at the values $t_{min}$,
$t_{mean}$, $t_{max}$, $mC_{min}$, etc. (twelve values for each
combination of $G_c$, $G_p$, and a mapping algorithm).

Next we form the geometric means of the twelve values over all graphs
in \walshawlarge and \complexnets, respectively. Thus we arrive at
twelve values $t^{gm}_{min}$, $\dots$ for any combination of a graph
class (\walshawlarge or \complexnets), a processor graph, and a
mapping algorithm. Finally, the last 9 values (all except runtimes) are
set into proportion to the corresponding values for $\initial$. This
yields the values $QmC^{gm}_{min}$, $QmC^{gm}_{mean}$,
$QmC^{gm}_{max}$, $QmD^{gm}_{min}$, $QmD^{gm}_{mean}$,
$QmD^{gm}_{max}$, $QaD^{gm}_{min}$, $QaD^{gm}_{mean}$ and
$QaD^{gm}_{max}$. A $Q$-value smaller than one means that the
quality is higher than that of $\initial$ because we are minimizing.

We also investigate the influence of graph partitioning on the quality
of the mapping algorithms. In addition to using \kahip as described
above, we apply two variants of \metis v. 5.1.0~\cite{Karypis13a}

\begin{enumerate}
\item We run $\gpmetis$ with the option of $k$-way partitioning, an
  allowed imbalance of $1.03$ and $20$ seeds (imbalance and 
  seed number are as for \kahip).
\item We run \ndmetis with $20$ seeds. This results in a fill-reducing
  ordering of $G_a$'s adjacency matrix. The ordering is then turned
  into a partitioning of $G_a$ by going through the vertices in the
  new order and assigning block numbers such that all blocks have
  almost equal size (maximal deviation is one vertex). We are aware
  that using $\ndmetis$ in this way is not a good choice in view of
  partitioning quality ($\ndmetis$ is made for other purposes). We
  proceed like this, however, \emph{because} we wish to test our
  collection of mapping algorithms on partitions with mediocre edge cut
  and MCV.
\end{enumerate}

We indicate the \metis-based graph partitioning that is underlying a
mapping algorithm by using the subscripts $\textsc{g}$ and
$\textsc{n}$ when employing $\gpmetis$ and $\ndmetis$,
respectively. As an example, $\greedyallcMO$ means that we applied
$\greedyallc$ to partitions obtained via $\ndmetis$.

Finally, for each $x \in \{t^{gm}_{min}, \dots\, QaD^{gm}_{max}\}$
(this set has 12 values), we form quotients $\mathcal{Q}x$ of the form

\begin{equation*}
\frac{\mbox{$x$ from \kahip partitions}}{\mbox{$x$ from \gpmetis
    partitions}} \mbox{~,~} \frac{\mbox{$x$ from \kahip
    partitions}}{\mbox{$x$ from \ndmetis partitions}}
\end{equation*}

As an example, $\mathcal{Q}aD^{gm}_{max} = 3.2098$ for $\greedyallcMO$
in Table~\ref{tab:app:meshesComp:2DGrid1} means that $aD^{gm}_{max}$
is worse by a factor of $3.2098$ if $\ndmetis$ is used instead of
$\kahip$ for mapping meshes onto 2DGrid($16 \times 16$).

%
%
\section{Results}
\label{sec:results}
%
\subsection{Mapping of Meshes onto Grids and Tori}
\label{sub:exp-grids-tori-meshes}

\begin{table}[htb]
  \caption{Performance of $\gpmetis$ and $\ndmetis$ on meshes compared
    to $\kahip$. Values smaller than one indicate that
    $\gpmetis\slash\ndmetis$ is faster or that the quality the
    $\gpmetis\slash\ndmetis$-partitions is higher.}
\begin{center}
  \begin{tabular}{ l | c c c | c c c | c c c}
           & $Time^{gm}_{min}$ & $Time^{gm}_{mean}$ &
    $Time^{gm}_{max}$ & $Cut^{gm}_{min}$ & $Cut^{gm}_{mean}$ &
    $Cut^{gm}_{max}$ & $MCV^{gm}_{min}$ & $MCV^{gm}_{mean}$ &
    $MCV^{gm}_{max}$ \\\hline \hline

    $\gpmetis$ & 0.0462 & 0.0451 & 0.0440 & 1.0101 & 1.0101 & 1.0121 & 0.9970 & 1.0449 & 1.1601\\
    $\ndmetis$ & 0.1026 & 0.0999 & 0.0985 & 2.2075 & 2.2371 & 2.2472 & 6.0976 & 5.9880 & 5.7471  
\end{tabular}
\end{center}
\label{tab:meansQuot}
\end{table}

Table~\ref{tab:meansQuot} shows a comparison of $\kahip$ partitions
with partitions from $\gpmetis$ and $\ndmetis$. We measure running
time, edge cut and MCV. As above, we record the best, mean and worst
result over $20$ seeds and calculate the geometric means of these
numbers over all meshes in our collection --- giving rise to the
numbers $Time^{gm}_{max}, \dots, MCV^{gm}_{max}$ in
Table~\ref{tab:meansQuot}. In terms of partition quality, $\gpmetis$
performs significantly poorer than $\kahip$ only in terms of
$MCV^{gm}_{mean}$ and $MCV^{gm}_{max}$. Here $\gpmetis$ is worse by
$4.49\%$ and $16.01\%$, respectively.
The partitions that we derived from $\ndmetis$ (in a deliberately
sub-optimal way) fall back drastically both in terms of the edge cut
and MCV. In particular, $Cut^{gm}_{mean} = 2.2371$ and
$MCV^{gm}_{mean} = 5.988$, which means that the edge cut from
$\ndmetis$ is more than double and that MCV increases almost six times
if $\ndmetis$ is used instead of $\kahip$.

Table~\ref{tab:app:meshes:2DTorus1} shows the quality of the mapping
algorithms for (i) partitions based on $\kahip$ and (ii) mapping onto
the $16 \times 16$ torus. Tables~I - V in the Appendix support the
following results.

\begin{table}[htb]
\caption{Mapping of meshes onto 2DTorus($16 \times 16$). Times
  $t^{gm}_{min}$, $t^{gm}_{mean}$ and $t^{gm}_{max }$ are in {\bf milliseconds}.}
\begin{center}
\scalebox{0.9}{
  \begin{tabular}{l || c c c | c  c  c | c  c  c  | c  c  c }
    Algo & $t^{gm}_{min}$ & $t^{gm}_{mean}$ & $t^{gm}_{max}$ & $QmC^{gm}_{min}$ &
    $QmC^{gm}_{mean}$ & $QmC^{gm}_{max}$ & $QmD^{gm}_{min}$ &
    $QmD^{gm}_{mean}$ & $QmD^{gm}_{max}$ & $QaD^{gm}_{min}$ &
    $QaD^{gm}_{mean}$ & $QaD^{gm}_{max}$\\ \hline
    $\random$     & 0.028 & 0.033 & 0.044 & 2.087 & 2.059 & 2.030 & 1.389 & 1.397 & 1.432 & 1.667 & 1.471 & 1.249\\
    $\rcm$        & 0.060 & 0.070 & 0.088 & 1.634 & 1.640 & 1.645 & 1.357 & 1.454 & 1.586 & 1.509 & 1.389 & 1.242\\
    $\durebi$     & 50.75 & 52.18 & 54.20 & 0.862 & 0.904 & 0.966 & 0.821 & 0.886 & 0.987 & 1.039 & 1.010 & 0.962\\
    $\greedyall$  & 0.948 & 0.971 & 0.997 & 1.310 & 1.316 & 1.324 & 1.252 & 1.291 & 1.369 & 1.359 & 1.263 & 1.164\\
    $\greedymin$  & 0.163 & 0.168 & 0.183 & 1.139 & 1.162 & 1.199 & 1.025 & 1.080 & 1.149 & 0.870 & 0.776 & {\bf 0.654}\\
    $\greedyallc$ & 0.918 & 0.953 & 0.987 & {\bf 0.683} & {\bf 0.707} & {\bf 0.736} & {\bf 0.665} & {\bf 0.706} & {\bf 0.766} & {\bf 0.730} & 0.780 & 0.871\\
    $\greedyminc$ & 0.869 & 0.939 & 0.100 & 0.793 & 0.813 & 0.844 & 0.739 & 0.789 & 0.849 & 0.745 & {\bf 0.756} & 0.780\\ \hline
  \end{tabular}
}
\end{center}
\label{tab:app:meshes:2DTorus1}
\end{table}

\begin{enumerate}
\item The mapping algorithms $\random$, $\rcm$ and $\greedyall$ are worse than
  $\initial$ on all accounts. While this was expected for $\random$, our data show
  that very simple mapping strategies are not worthwhile if the underlying
  partition is good.

\item The algorithm $\greedymin$ beats $\initial$ only in terms of
  average dilation. The improvement is, however, a major one in some
  cases, \eg $QaD^{gm}_{mean} = 0.776$ and $QaD^{gm}_{max} = 0.654$
  for the $16 \times 16$ 2D torus (see
  Table~\ref{tab:app:meshes:2DTorus1}). Another strong point of
  $\greedymin$ is its low running time.

\item On all six processor graphs our new mapping algorithm
  $\greedyallc$ yields the best maximum congestion, $mC$, and the best
  maximum dilation, $mD$. This holds not only for the (geometric mean
  over all meshes of the) average over all seeds, but also if the best
  or the worst result is taken over all seeds. The quotients are
  between $0.556$ and $0.789$. In terms of running time, we are
  in-between that of $\greedymin$ and $\durebi$.

\item $\durebi$ yields many major improvements over $\initial$ and,
  discarding average dilation, is worse only once (in terms of
  $QmD^{gm}_{max}$ on the 3D torus, see Table~V in the
  Appendix). $\durebi$ often comes close to $\greedyallc$ and
  sometimes beats it on average dilation.

\item $\greedyminc$ has its strengths on tori and often beats
  $\greedyallc$ on average dilation (on grids \emph{and}
  tori). Interestingly, the overall quality of $\greedyall$ is much
  worse than that of $\greedymin$ (both from previous work),
  while this trend is reversed if we
  look at the modified versions $\greedyallc$ and $\greedyminc$.
\end{enumerate}

We now look at the influence of the partitioning quality on the
quality of the mapping algorithms (see
Table~\ref{tab:app:meshesComp:2DGrid1} (Table~VI in the Appendix
provides more evidence). As for $\kahip$ vs. $\gpmetis$, the small
lead of $\kahip$ over $\gpmetis$ \wrt MCV translates into an even
smaller lead of the corresponding mappings. Moreover, this small lead
is only on average, and there are cases where $\gpmetis$ partitions
lead to better mapping results.  As for $\kahip$ vs. $\ndmetis$, poor
edge cut and$\slash$or MCV seem to have a deteriorating effect on
mapping quality.

\begin{table}[htb]
\caption{Mapping of meshes onto 2DGrid($16 \times 16$).}
\begin{center}
\scalebox{0.85}{
  \begin{tabular}{l || c c c | c  c  c | c  c  c  | c  c  c }
    Algo & $\mathcal{Q}t^{gm}_{min}$ & $\mathcal{Q}t^{gm}_{mean}$ & $\mathcal{Q}t^{gm}_{max}$ & $\mathcal{Q}mC^{gm}_{min}$ &
    $\mathcal{Q}mC^{gm}_{mean}$ & $\mathcal{Q}mC^{gm}_{max}$ & $\mathcal{Q}mD^{gm}_{min}$ &
    $\mathcal{Q}mD^{gm}_{mean}$ & $\mathcal{Q}mD^{gm}_{max}$ & $\mathcal{Q}aD^{gm}_{min}$ &
    $\mathcal{Q}aD^{gm}_{mean}$ & $\mathcal{Q}aD^{gm}_{max}$\\ \hline
    $\initialM$     & 1.0205 & 1.0152 & 1.1762 & 0.9930 & 0.9963 & 0.9881 & 1.0010 & 1.0002 & 0.9985 & 1.0042 & 1.0071 & 0.9690\\
    $\initialMO$    & 1.0346 & 1.0506 & 1.1932 & 1.9068 & 1.9594 & 2.0013 & 2.7783 & 2.8391 & 2.8141 & 3.9285 & 4.2721 & 4.2889\\
    $\randomM$      & 0.9944 & 0.9749 & 0.9809 & 0.9961 & 0.9983 & 1.0027 & 0.9968 & 0.9847 & 0.9431 & 1.0251 & 1.0307 & 1.0524\\
    $\randomMO$     & 1.0120 & 0.9887 & 0.9510 & 1.9659 & 2.0315 & 2.0960 & 2.5159 & 2.5489 & 2.5157 & 3.8256 & 3.9547 & 4.3704\\
    $\rcmM$         & 1.0287 & 1.0087 & 1.0103 & 0.9980 & 1.0020 & 0.9955 & 1.0033 & 1.0169 & 1.0310 & 1.0473 & 1.0225 & 1.0019\\
    $\rcmMO$        & 1.1701 & 1.1362 & 1.1554 & 2.0780 & 2.1376 & 2.1856 & 2.8376 & 2.8998 & 3.0228 & 3.6763 & 3.8054 & 4.2111\\
    $\durebiM$      & 1.0083 & 1.0034 & 0.9956 & 0.9991 & 1.0096 & 1.0258 & 0.9951 & 1.0112 & 1.0131 & 1.0179 & 1.0441 & 0.9811\\
    $\durebiMO$     & 0.9980 & 1.0119 & 1.0194 & 1.9035 & 2.0223 & 1.9517 & 2.6668 & 2.8051 & 2.7540 & 2.5027 & 2.9897 & 3.1595\\
    $\greedyminM$   & 1.0039 & 1.0021 & 1.0182 & 1.0110 & 0.9953 & 0.9948 & 1.0295 & 0.9959 & 0.9923 & 1.0031 & 1.0088 & 1.0034\\
    $\greedyminMO$  & 1.0095 & 1.0177 & 1.1114 & 1.5176 & 1.4785 & 1.4779 & 2.5373 & 2.5238 & 2.5750 & 1.5832 & 1.5819 & 1.6219\\
    $\greedyallcM$  & 1.0091 & 1.0074 & 1.0236 & 0.9999 & 1.0135 & 1.0144 & 1.0491 & 1.0200 & 1.0085 & 1.0153 & 0.9919 & 0.9771\\
    $\greedyallcMO$ & 1.1313 & 1.1499 & 1.1819 & 1.9840 & 1.9265 & 1.8896 & 2.9915 & 2.9117 & 2.9811 & 2.1990 & 2.6233 & 3.2098\\
    $\greedymincM$  & 1.0002 & 1.0073 & 1.0205 & 1.0121 & 0.9887 & 0.9809 & 1.0166 & 0.9987 & 0.9979 & 1.0018 & 1.0013 & 0.9449\\
    $\greedymincMO$ & 1.1869 & 1.1879 & 1.1993 & 1.6175 & 1.5741 & 1.5145 & 2.5819 & 2.7487 & 2.8079 & 1.6267 & 2.2717 & 3.5661\\ \hline
  \end{tabular}
}
\end{center}
\label{tab:app:meshesComp:2DGrid1}
\end{table}

%
\subsection{Mapping of Complex Networks onto Grids and Tori.}
\label{sub:exp-complex-grids-tori}

Table~\ref{tab:meansQuotComplex} shows a comparison of
$\kahip$ partitions with partitions from $\gpmetis$ and
$\ndmetis$. For a description of the table see the explanation of
Table~\ref{tab:meansQuot} in Section~\ref{sub:exp-grids-tori-meshes}.

\begin{table}[htb]
  \caption{Performance of $\gpmetis$ and $\ndmetis$ on complex
    networks compared to $\kahip$. Values smaller than one indicate
    that $\gpmetis\slash\ndmetis$ is faster or that the quality the
    $\gpmetis\slash\ndmetis$-partitions is higher.}
\begin{center}
  \begin{tabular}{ l | c c c | c c c | c c c}
           & $Time^{gm}_{min}$ & $Time^{gm}_{mean}$ &
    $Time^{gm}_{max}$ & $Cut^{gm}_{min}$ & $Cut^{gm}_{mean}$ &
    $Cut^{gm}_{max}$ & $MCV^{gm}_{min}$ & $MCV^{gm}_{mean}$ &
    $MCV^{gm}_{max}$ \\\hline \hline
    $\gpmetis$ & 0.0083 & 0.0081 & 0.0078 & 1.0634 & 1.0619 & 1.0560 & 1.2531 & 1.2066 & 1.1536\\
    $\ndmetis$ & 0.0262 & 0.0268 & 0.0257 & 2.0284 & 2.0202 & 2.0121 & 1.8416 & 1.9157 & 2.0040
\end{tabular}
\end{center}
\label{tab:meansQuotComplex}
\end{table}

Compared to the picture we saw on meshes, $\kahip$ now also leads in
terms of the edge cut. Moreover, the lead of $\kahip$ in terms of MCV
compared to $\gpmetis$ and $\ndmetis$ is even more pronounced (about
$20\%$).

Regarding topology mapping based on $\kahip$ partitions, we only
comment on results that deviate from those that we have described for
meshes (especially running times show the same trends). The main
differences are in the maximum and average dilation. Sometimes $\rcm$
and even $\random$ yield even lower maximum dilation than
$\greedyallc$. Moreover, average dilation behaves quite erratically,
as is revealed by a comparison between the $aD$-values of
$\greedyminc$ in Table~\ref{tab:app:social:2DTorus1} and
Tables~VII through~XI in the Appendix.

\begin{table}[htb]
\caption{Mapping of complex networks onto 2DTorus($16 \times 16$). Times
  $t^{gm}_{min}$, $t^{gm}_{mean}$ and $t^{gm}_{max }$ are in {\bf milliseconds}.}
\begin{center}
\scalebox{0.9}{
  \begin{tabular}{l || c c c | c  c  c | c  c  c  | c  c  c }
    Algo & $t^{gm}_{min}$ & $t^{gm}_{mean}$ & $t^{gm}_{max}$ & $QmC^{gm}_{min}$ &
    $QmC^{gm}_{mean}$ & $QmC^{gm}_{max}$ & $QmD^{gm}_{min}$ &
    $QmD^{gm}_{mean}$ & $QmD^{gm}_{max}$ & $QaD^{gm}_{min}$ &
    $QaD^{gm}_{mean}$ & $QaD^{gm}_{max}$\\ \hline
    $\random$     & 0.028 & 0.032 & 0.042 & 1.511 & 1.509 & 1.513 & {\bf 0.777} & {\bf 0.780} & 0.810 & 3.291 & 3.241 & 2.840\\
    $\rcm$        & 0.104 & 0.122 & 0.161 & 1.366 & 1.416 & 1.455 & 0.822 & 0.868 & 0.931 & 2.672 & 2.919 & 2.699\\
    $\durebi$     & 124.8 & 138.5 & 154.0 & 0.982 & 1.003 & 1.021 & 0.853 & 0.876 & 0.926 & 1.084 & 1.250 & 1.476\\
    $\greedyall$  & 5.182 & 5.344 & 5.684 & 1.100 & 1.119 & 1.131 & 1.068 & 1.011 & 0.985 & 1.189 & 1.315 & 1.301\\
    $\greedymin$  & 0.248 & 0.258 & 0.292 & 1.057 & 1.054 & 1.056 & 1.109 & 1.056 & 1.015 & 0.858 & 0.676 & 0.496\\
    $\greedyallc$ & 5.647 & 5.871 & 6.268 & {\bf 0.841} & {\bf 0.839} & {\bf 0.839} & 0.863 & 0.820 & {\bf 0.801} & {\bf 0.531} & {\bf 0.442} & {\bf 0.351}\\
    $\greedyminc$ & 5.251 & 5.550 & 6.153 & 0.858 & 0.856 & 0.855 & 0.857 & 0.829 & 0.808 & 0.644 & 0.575 & 0.481\\ \hline
  \end{tabular}
}
\end{center}
\label{tab:app:social:2DTorus1}
\end{table}

Regarding maximum congestion, $mC$, the picture is the same as we saw
for meshes: Our new algorithm $\greedyallc$ always yields the best results.

Regarding the influence of partitioning quality on the quality of the
mappings we see that the higher partitioning quality of $\kahip$
compared to $\gpmetis$ (in terms of the edge cut and MCV) does not
translate into considerably better mappings, see
Table~\ref{tab:app:socialComp:2DGrid1} (for additional evidence see
Table~XII in the Appendix). As in the case of meshes, the
partitions that we derived from $\ndmetis$ (in a deliberately
sub-optimal way) lead to poor mappings.

\begin{table}[htb]
\caption{Mapping of complex networks onto 2DGrid($16 \times 16$).}
\begin{center}
\scalebox{0.85}{
  \begin{tabular}{l || c c c | c  c  c | c  c  c  | c  c  c }
    Algo & $\mathcal{Q}t^{gm}_{min}$ & $\mathcal{Q}t^{gm}_{mean}$ & $\mathcal{Q}t^{gm}_{max}$ & $\mathcal{Q}mC^{gm}_{min}$ &
    $\mathcal{Q}mC^{gm}_{mean}$ & $\mathcal{Q}mC^{gm}_{max}$ & $\mathcal{Q}mD^{gm}_{min}$ &
    $\mathcal{Q}mD^{gm}_{mean}$ & $\mathcal{Q}mD^{gm}_{max}$ & $\mathcal{Q}aD^{gm}_{min}$ &
    $\mathcal{Q}aD^{gm}_{mean}$ & $\mathcal{Q}aD^{gm}_{max}$\\ \hline
    $\initialM$     & 0.9723 & 0.9838 & 0.9812 & 1.0243 & 1.0407 & 1.0623 & 0.9914 & 0.9955 & 1.0209 & 1.2791 & 1.4033 & 1.4395\\
    $\initialMO$    & 0.9684 & 0.9963 & 0.9832 & 10.004 & 10.146 & 10.195 & 2.9752 & 2.8494 & 2.8637 & 3.9055 & 3.7684 & 3.5497\\
    $\randomM$      & 1.0292 & 1.0234 & 1.0285 & 0.9832 & 0.9910 & 0.9983 & 1.0866 & 1.1270 & 1.1680 & 1.0194 & 1.0167 & 1.1733\\
    $\randomMO$     & 1.0127 & 1.0272 & 1.0027 & 7.5795 & 7.7739 & 7.9801 & 2.7521 & 2.8475 & 2.9019 & 1.8046 & 1.5884 & 1.6585\\
    $\rcmM$         & 1.0970 & 1.0851 & 0.9894 & 1.0087 & 0.9939 & 0.9991 & 1.1129 & 1.1367 & 1.1333 & 0.9571 & 0.9581 & 1.0146\\
    $\rcmMO$        & 0.9611 & 0.9963 & 0.9008 & 7.5698 & 7.8690 & 8.2057 & 3.2024 & 3.2317 & 2.9598 & 1.6997 & 1.6083 & 1.4822\\
    $\durebiM$      & 1.0486 & 1.0492 & 1.0620 & 1.0076 & 1.0147 & 1.0027 & 1.0624 & 1.0745 & 1.0913 & 0.9775 & 1.0832 & 1.0793\\
    $\durebiMO$     & 0.4242 & 0.4094 & 0.3885 & 7.6865 & 8.1313 & 8.5734 & 3.4637 & 3.7612 & 3.9759 & 2.0400 & 2.1615 & 1.7987\\
    $\greedyminM$   & 1.0375 & 1.0341 & 1.0102 & 1.0284 & 1.0267 & 1.0278 & 1.0060 & 1.0440 & 1.0791 & 1.0846 & 1.1068 & 1.0990\\
    $\greedyminMO$  & 0.6850 & 0.6860 & 0.7321 & 9.0648 & 9.1580 & 9.4715 & 4.4637 & 4.3656 & 4.4046 & 2.5795 & 2.5664 & 2.5504\\
    $\greedyallcM$  & 1.0637 & 1.0637 & 1.0608 & 1.0450 & 1.0392 & 1.0392 & 1.0527 & 1.0582 & 1.0594 & 1.3140 & 1.0434 & 0.9078\\
    $\greedyallcMO$ & 0.3690 & 0.3700 & 0.3559 & 7.9916 & 8.0772 & 8.1791 & 3.5868 & 3.5635 & 3.6554 & 4.0989 & 2.5944 & 1.6957\\
    $\greedymincM$  & 1.0847 & 1.0731 & 1.0698 & 1.0454 & 1.0411 & 1.0347 & 1.1240 & 1.1579 & 1.2136 & 1.0790 & 1.0542 & 0.9809\\
    $\greedymincMO$ & 0.3688 & 0.3721 & 0.3769 & 9.2443 & 9.3823 & 9.8506 & 5.2981 & 5.1511 & 5.1905 & 2.4001 & 1.7601 & 1.3172\\ \hline
  \end{tabular}
}
\end{center}
\label{tab:app:socialComp:2DGrid1}
\end{table}

\FloatBarrier

%
%
\section{Conclusions and Future Work}
\label{sec:conclusions}
We performed extensive static mapping experiments, our scenario being
a consecutive pipeline of graph partitioning and bijective topology mapping.
These experiments involved  two classes of application
graphs (8 meshes, 12 complex networks), three ways to partition the
application graphs (one by $\kahip$, two by \metis), six
processor graphs (3 grids, 3 tori) and 8 mapping algorithms.

Our results indicate that the strengths and weaknesses of the mapping algorithms are, to
a large extent, independent of the class of application graphs (mesh
or complex network) and the processor graphs. The main differences are
in the maximum and average dilation. Especially the latter behaves
erratically in the case of complex networks.

Second, the quality of the partitions, both in terms of edge cut and
MCV, has little influence on the quality of the mapping, except in
cases where MCV is very poor. Thus, even MCV is not a good indicator
of how well a partition can be mapped onto a processor graph --- at
least within the realm of our experiments.

Third, our variant of a greedy mapping algorithm by Brandfass \etal,
\ie $\greedyallc$, clearly dominates all state-of-the art algorithms we considered
in terms of maximum congestion. The running time of our algorithm is $\bigO(\vert V_c
\vert \vert E_c \vert)$, where $V_c$ and $E_c$ is the vertex and the
edge set of the communication graph, respectively (and therefore usually fairly small).

If the weak influence of partition quality on mapping quality is
affirmed for more classes of application graphs and more parallel
architectures, improvements of static mapping
are likely to come only out of new combinations of
partitioning and mapping. In the future we will investigate how 
to minimize the communication volume specified in Equation~\ref{eq+}
by such a coupled approach.


\bibliographystyle{IEEEtran}
\bibliography{roland,refs-parco,paper2.bib}

\section*{Appendix}
%
%
\begin{table*}[!h]
\caption{Mapping of meshes onto 2DGrid($16 \times 16$). Times
  $t^{gm}_{min}$, $t^{gm}_{mean}$ and $t^{gm}_{max }$ are in {\bf milliseconds}.}
\begin{center}
\scalebox{0.9}{
  \begin{tabular}{l || c c c | c  c  c | c  c  c  | c  c  c }
    Algo & $t^{gm}_{min}$ & $t^{gm}_{mean}$ & $t^{gm}_{max}$ & $QmC^{gm}_{min}$ &
    $QmC^{gm}_{mean}$ & $QmC^{gm}_{max}$ & $QmD^{gm}_{min}$ &
    $QmD^{gm}_{mean}$ & $QmD^{gm}_{max}$ & $QaD^{gm}_{min}$ &
    $QaD^{gm}_{mean}$ & $QaD^{gm}_{max}$\\ \hline
    $\random$     & 0.029 & 0.033 & 0.040 & 2.295 & 2.264 & 2.243 & 1.941 & 1.926 & 1.861 & 1.744 & 1.550 & 1.423\\
    $\rcm$        & 0.057 & 0.068 & 0.079 & 1.564 & 1.563 & 1.572 & 1.288 & 1.297 & 1.304 & 1.670 & 1.530 & 1.393\\
    $\durebi$     & 49.40 & 50.82 & 52.68 & 0.756 & 0.818 & 0.881 & 0.697 & 0.750 & 0.773 & 0.913 & 0.974 & 1.012\\
    $\greedyall$  & 0.951 & 0.973 & 0.992 & 1.679 & 1.680 & 1.678 & 1.359 & 1.361 & 1.292 & 1.782 & 1.720 & 1.579\\
    $\greedymin$  & 0.160 & 0.164 & 0.177 & 1.120 & 1.192 & 1.274 & 1.092 & 1.136 & 1.175 & 0.871 & 0.803 & {\bf 0.714}\\
    $\greedyallc$ & 0.906 & 0.942 & 0.995 & {\bf 0.665} & {\bf 0.722} & {\bf 0.789} & {\bf 0.626} & {\bf 0.665} & {\bf 0.736} & 0.882 & 1.027 & 1.253\\
    $\greedyminc$ & 0.828 & 0.875 & 0.919 & 0.817 & 0.890 & 0.957 & 0.750 & 0.785 & 0.817 & {\bf 0.787} & {\bf 0.792} & 0.901\\ \hline
  \end{tabular}
}
\end{center}
\label{tab:app:meshes:2DGrid1}
\end{table*}

\begin{table*}[!h]
\caption{Mapping of meshes onto 2DGrid($32 \times 32$). Times
  $t^{gm}_{min}$, $t^{gm}_{mean}$ and $t^{gm}_{max }$ are in {\bf milliseconds}.}
\begin{center}
\scalebox{0.9}{
  \begin{tabular}{l || c c c | c  c  c | c  c  c  | c  c  c }
    Algo & $t^{gm}_{min}$ & $t^{gm}_{mean}$ & $t^{gm}_{max}$ & $QmC^{gm}_{min}$ &
    $QmC^{gm}_{mean}$ & $QmC^{gm}_{max}$ & $QmD^{gm}_{min}$ &
    $QmD^{gm}_{mean}$ & $QmD^{gm}_{max}$ & $QaD^{gm}_{min}$ &
    $QaD^{gm}_{mean}$ & $QaD^{gm}_{max}$\\ \hline
    $\random$     & 0.093 & 0.110 & 0.135 & 2.933 & 2.871 & 2.806 & 2.268 & 2.169 & 2.084 & 1.769 & 1.637 & 1.440\\
    $\rcm$        & 0.205 & 0.239 & 0.273 & 1.776 & 1.782 & 1.781 & 1.441 & 1.433 & 1.418 & 1.784 & 1.683 & 1.481\\
    $\durebi$     & 216.5 & 221.6 & 227.7 & 0.715 & 0.757 & 0.806 & 0.655 & 0.690 & 0.756 & 0.952 & 1.060 & 1.125\\
    $\greedyall$  & 17.55 & 18.53 & 18.97 & 2.116 & 2.084 & 2.068 & 1.609 & 1.565 & 1.522 & 1.956 & 1.838 & 1.632\\
    $\greedymin$  & 3.229 & 3.867 & 3.947 & 1.294 & 1.371 & 1.442 & 1.289 & 1.320 & 1.408 & 0.936 & {\bf 0.858} & {\bf 0.732}\\
    $\greedyallc$ & 17.20 & 18.05 & 18.89 & {\bf 0.569} & {\bf 0.626} & {\bf 0.681} & {\bf 0.556} & {\bf 0.615} & {\bf 0.689} & 0.876 & 1.137 & 1.381\\
    $\greedyminc$ & 15.98 & 16.71 & 17.70 & 0.887 & 0.948 & 1.040 & 0.770 & 0.837 & 0.916 & {\bf 0.869} & 0.892 & 0.966\\ \hline
  \end{tabular}
}
\end{center}
\label{tab:meshes:2DGrid2}
\end{table*}

\begin{table*}[!h]
\caption{Mapping of meshes onto 3DGrid($8 \times 8 \times 8$). Times
  $t^{gm}_{min}$, $t^{gm}_{mean}$ and $t^{gm}_{max }$ are in {\bf milliseconds}.}
\begin{center}
\scalebox{0.9}{
  \begin{tabular}{l || c c c | c  c  c | c  c  c  | c  c  c }
    Algo & $t^{gm}_{min}$ & $t^{gm}_{mean}$ & $t^{gm}_{max}$ & $QmC^{gm}_{min}$ &
    $QmC^{gm}_{mean}$ & $QmC^{gm}_{max}$ & $QmD^{gm}_{min}$ &
    $QmD^{gm}_{mean}$ & $QmD^{gm}_{max}$ & $QaD^{gm}_{min}$ &
    $QaD^{gm}_{mean}$ & $QaD^{gm}_{max}$\\ \hline
    $\random$     & 0.047 & 0.056 & 0.070 & 2.169 & 2.131 & 2.094 & 1.749 & 1.739 & 1.690 & 1.664 & 1.508 & 1.329\\
    $\rcm$        & 0.113 & 0.129 & 0.148 & 1.631 & 1.622 & 1.616 & 1.289 & 1.314 & 1.305 & 1.476 & 1.372 & 1.216\\
    $\durebi$     & 112.6 & 115.6 & 119.8 & 0.814 & 0.856 & 0.920 & 0.719 & 0.769 & 0.819 & 0.935 & 0.995 & 0.946\\
    $\greedyall$  & 4.006 & 4.050 & 4.101 & 1.723 & 1.704 & 1.688 & 1.348 & 1.326 & 1.267 & 1.670 & 1.618 & 1.470\\
    $\greedymin$  & 0.687 & 0.708 & 0.722 & 1.214 & 1.232 & 1.255 & 1.051 & 1.061 & 1.082 & 0.976 & {\bf 0.887} & {\bf 0.780}\\
    $\greedyallc$ & 3.875 & 3.996 & 4.104 & {\bf 0.683} & {\bf 0.713} & {\bf 0.736} & {\bf 0.617} & {\bf 0.633} & {\bf 0.638} & {\bf 0.876} & 1.055 & 1.151\\
    $\greedyminc$ & 3.683 & 3.871 & 4.067 & 0.827 & 0.859 & 0.891 & 0.695 & 0.718 & 0.757 & 1.268 & 1.339 & 1.388\\ \hline
  \end{tabular}
}
\end{center}
\label{tab:meshes:3DGrid}
\end{table*}

\begin{table*}[!h]
\caption{Mapping of meshes onto 2DTorus($32 \times 32$). Times
  $t^{gm}_{min}$, $t^{gm}_{mean}$ and $t^{gm}_{max }$ are in {\bf milliseconds}.}
\begin{center}
\scalebox{0.9}{
  \begin{tabular}{l || c c c | c  c  c | c  c  c  | c  c  c }
    Algo & $t^{gm}_{min}$ & $t^{gm}_{mean}$ & $t^{gm}_{max}$ & $QmC^{gm}_{min}$ &
    $QmC^{gm}_{mean}$ & $QmC^{gm}_{max}$ & $QmD^{gm}_{min}$ &
    $QmD^{gm}_{mean}$ & $QmD^{gm}_{max}$ & $QaD^{gm}_{min}$ &
    $QaD^{gm}_{mean}$ & $QaD^{gm}_{max}$\\ \hline
    $\random$     & 0.090 & 0.111 & 0.139 & 2.656 & 2.609 & 2.566 & 1.487 & 1.4728 & 1.444 & 1.661 & 1.491 & 1.369\\
    $\rcm$        & 0.213 & 0.249 & 0.295 & 1.942 & 1.942 & 1.932 & 1.480 & 1.5121 & 1.552 & 1.536 & 1.400 & 1.312\\
    $\durebi$     & 212.9 & 217.0 & 221.6 & 0.794 & 0.843 & 0.896 & 0.782 & 0.8642 & 0.938 & 1.059 & 1.068 & 1.111\\
    $\greedyall$  & 17.79 & 18.07 & 18.42 & 1.529 & 1.526 & 1.533 & 1.638 & 1.7147 & 1.765 & 1.407 & 1.319 & 1.215\\
    $\greedymin$  & 4.031 & 4.068 & 4.110 & 1.303 & 1.336 & 1.360 & 1.166 & 1.2162 & 1.259 & 0.892 & {\bf 0.787} & {\bf 0.683}\\
    $\greedyallc$ & 17.53 & 18.08 & 18.94 & {\bf 0.569} & {\bf 0.611} & {\bf 0.647} & {\bf 0.609} & {\bf 0.6843} & {\bf 0.752} & {\bf 0.726} & 0.815 & 0.899\\
    $\greedyminc$ & 16.93 & 17.68 & 18.30 & 0.778 & 0.820 & 0.859 & 0.752 & 0.8148 & 0.886 & 0.892 & 0.973 & 1.028\\ \hline
  \end{tabular}
}
\end{center}
\label{tab:meshes:2DTorus2}
\end{table*}

\begin{table*}[!h]
\caption{Mapping of meshes onto 3DTorus($8 \times 8 \times 8$). Times
  $t^{gm}_{min}$, $t^{gm}_{mean}$ and $t^{gm}_{max }$ are in {\bf milliseconds}.}
\begin{center}
\scalebox{0.9}{
  \begin{tabular}{l || c c c | c  c  c | c  c  c  | c  c  c }
    Algo & $t^{gm}_{min}$ & $t^{gm}_{mean}$ & $t^{gm}_{max}$ & $QmC^{gm}_{min}$ &
    $QmC^{gm}_{mean}$ & $QmC^{gm}_{max}$ & $QmD^{gm}_{min}$ &
    $QmD^{gm}_{mean}$ & $QmD^{gm}_{max}$ & $QaD^{gm}_{min}$ &
    $QaD^{gm}_{mean}$ & $QaD^{gm}_{max}$\\ \hline
    $\random$     & 0.049 & 0.058 & 0.075 & 2.052 & 2.013 & 2.000 & 1.307 & 1.331 & 1.337 & 1.571 & 1.447 & 1.392\\
    $\rcm$        & 0.115 & 0.131 & 0.154 & 1.638 & 1.637 & 1.641 & 1.317 & 1.395 & 1.502 & 1.485 & 1.378 & 1.360\\
    $\durebi$     & 115.8 & 118.5 & 121.2 & 0.931 & 0.972 & 1.049 & 0.859 & 0.943 & 1.039 & 1.057 & 1.028 & 1.045\\
    $\greedyall$  & 3.848 & 3.910 & 4.028 & 1.303 & 1.230 & 1.312 & 1.161 & 1.214 & 1.221 & 1.262 & 1.196 & 1.173\\
    $\greedymin$  & 0.656 & 0.670 & 0.684 & 1.209 & 1.217 & 1.240 & 1.002 & 1.049 & 1.073 & 0.903 & 0.797 & {\bf 0.726}\\
    $\greedyallc$ & 3.882 & 4.050 & 4.178 & {\bf 0.751} & {\bf 0.757} & {\bf 0.767} & {\bf 0.683} & {\bf 0.719} & {\bf 0.711} & 0.755 & 0.806 & 0.913\\
    $\greedyminc$ & 3.897 & 4.125 & 4.463 & 0.840 & 0.842 & 0.858 & 0.727 & 0.761 & 0.790 & {\bf 0.739} & {\bf 0.743} & 0.783\\ \hline
  \end{tabular}
}
\end{center}
\label{tab:app:meshes:3DTorus}
\end{table*}

\begin{table*}[!h]
\caption{Mapping of meshes onto 2DTorus($16 \times 16$).}
\begin{center}
\scalebox{0.85}{
  \begin{tabular}{l || c c c | c  c  c | c  c  c  | c  c  c }
    Algo & $\mathcal{Q}t^{gm}_{min}$ & $\mathcal{Q}t^{gm}_{mean}$ & $\mathcal{Q}t^{gm}_{max}$ & $\mathcal{Q}mC^{gm}_{min}$ &
    $\mathcal{Q}mC^{gm}_{mean}$ & $\mathcal{Q}mC^{gm}_{max}$ & $\mathcal{Q}mD^{gm}_{min}$ &
    $\mathcal{Q}mD^{gm}_{mean}$ & $\mathcal{Q}mD^{gm}_{max}$ & $\mathcal{Q}aD^{gm}_{min}$ &
    $\mathcal{Q}aD^{gm}_{mean}$ & $\mathcal{Q}aD^{gm}_{max}$\\ \hline
    $\initialM$     & 1.0361 & 0.9916 & 0.8434 & 0.9987 & 0.9964 & 0.9968 & 1.0042 & 1.0032 & 0.9917 & 0.9844 & 1.0209 & 1.0812\\
    $\initialMO$    & 0.9799 & 0.9805 & 0.8702 & 1.6731 & 1.6923 & 1.7122 & 2.5803 & 2.7067 & 2.8812 & 1.5955 & 1.5739 & 1.5854\\
    $\randomM$      & 1.0016 & 1.0165 & 0.9776 & 1.0019 & 1.0014 & 1.0027 & 1.0179 & 1.0257 & 1.0847 & 1.0188 & 1.0256 & 1.0516\\
    $\randomMO$     & 1.0269 & 1.0453 & 1.0364 & 1.9834 & 2.0339 & 2.0892 & 2.6819 & 2.7297 & 2.9002 & 3.8130 & 4.0997 & 4.3485\\
    $\rcmM$         & 1.0212 & 1.0093 & 1.1226 & 1.0140 & 1.0005 & 1.0014 & 1.0219 & 1.0256 & 0.9761 & 1.0133 & 1.0188 & 1.0131\\
    $\rcmMO$        & 1.1623 & 1.1217 & 1.1519 & 2.0952 & 2.1371 & 2.1744 & 2.7086 & 2.7540 & 2.7208 & 3.7498 & 4.0711 & 4.3358\\
    $\durebiM$      & 1.0058 & 0.9997 & 0.9919 & 1.0062 & 1.0070 & 1.0161 & 1.0041 & 1.0061 & 1.0642 & 1.0238 & 1.0044 & 1.0358\\
    $\durebiMO$     & 0.9996 & 1.0072 & 1.0058 & 1.9550 & 2.0459 & 1.9538 & 2.8690 & 2.9701 & 3.1304 & 2.5652 & 3.1615 & 3.8322\\
    $\greedyminM$   & 1.0065 & 1.0013 & 1.0038 & 1.0162 & 0.9976 & 1.0011 & 0.9978 & 0.9865 & 0.9606 & 1.0100 & 1.0115 & 1.0202\\
    $\greedyminMO$  & 1.0185 & 1.0110 & 1.0252 & 1.5790 & 1.5512 & 1.5467 & 2.4673 & 2.5399 & 2.5795 & 1.6415 & 1.6700 & 1.7920\\
    $\greedyallcM$  & 1.0063 & 1.0052 & 1.0014 & 0.9941 & 0.9996 & 0.9941 & 1.0085 & 0.9930 & 0.9828 & 1.0682 & 1.0244 & 0.9929\\
    $\greedyallcMO$ & 1.1169 & 1.1379 & 1.1481 & 2.0367 & 1.9924 & 1.9810 & 2.8488 & 3.0034 & 3.0069 & 2.0127 & 2.3230 & 2.9728\\
    $\greedymincM$  & 1.0059 & 1.0101 & 1.0793 & 0.9915 & 0.9928 & 0.9991 & 0.9775 & 0.9935 & 1.0491 & 1.0182 & 1.0197 & 1.0345\\
    $\greedymincMO$ & 1.1779 & 1.1475 & 1.1270 & 1.7459 & 1.7276 & 1.6917 & 2.4912 & 2.6535 & 2.7047 & 1.7459 & 2.4074 & 3.3671\\ \hline
  \end{tabular}
}
\end{center}
\label{tab:app:meshesComp:2DTorus1}
\end{table*}

%
%

\begin{table*}[!h]
\caption{Mapping of complex networks onto 2DGrid($16 \times 16$). Times
  $t^{gm}_{min}$, $t^{gm}_{mean}$ and $t^{gm}_{max }$ are in {\bf milliseconds}.}
\begin{center}
\scalebox{0.9}{
  \begin{tabular}{l || c c c | c  c  c | c  c  c  | c  c  c }
    Algo & $t^{gm}_{min}$ & $t^{gm}_{mean}$ & $t^{gm}_{max}$ & $QmC^{gm}_{min}$ &
    $QmC^{gm}_{mean}$ & $QmC^{gm}_{max}$ & $QmD^{gm}_{min}$ &
    $QmD^{gm}_{mean}$ & $QmD^{gm}_{max}$ & $QaD^{gm}_{min}$ &
    $QaD^{gm}_{mean}$ & $QaD^{gm}_{max}$\\ \hline
    $\random$     & 0.028 & 0.032 & 0.042 & 1.593 & 1.597 & 1.597 & 0.987 & 0.925 & 0.943 & 3.737 & 3.831 & 3.361\\
    $\rcm$        & 0.101 & 0.120 & 0.155 & 1.362 & 1.419 & 1.457 & 0.968 & 0.956 & 1.032 & 2.783 & 3.267 & 3.411\\
    $\durebi$     & 124.8 & 138.4 & 154.8 & 0.937 & 0.957 & 0.989 & 0.762 & 0.750 & 0.778 & 1.012 & 1.154 & 1.562\\
    $\greedyall$  & 5.176 & 5.337 & 5.610 & 1.078 & 1.098 & 1.104 & 1.096 & 0.975 & 0.925 & 1.365 & 1.625 & 1.630\\
    $\greedymin$  & 0.245 & 0.256 & 0.230 & 1.043 & 1.045 & 1.038 & 1.003 & 0.929 & 0.897 & 0.774 & 0.627 & {\bf 0.453}\\
    $\greedyallc$ & 5.622 & 5.837 & 6.173 & {\bf 0.799} & {\bf 0.813} & {\bf 0.827} & 0.798 & 0.738 & 0.730 & 0.847 & 1.197 & 1.461\\
    $\greedyminc$ & 5.243 & 5.488 & 5.836 & 0.849 & 0.854 & 0.856 & {\bf 0.711} & {\bf 0.674} & {\bf 0.669} & {\bf 0.554} & {\bf 0.548} & 0.557\\ \hline
  \end{tabular}
}
\end{center}
\label{tab:app:social:2DGrid1}
\end{table*}

\begin{table*}[!h]
\caption{Mapping of complex networks onto 2DGrid($32 \times 32$). Times
  $t^{gm}_{min}$, $t^{gm}_{mean}$ and $t^{gm}_{max }$ are in {\bf milliseconds}.}
\begin{center}
\scalebox{0.9}{
  \begin{tabular}{l || c c c | c  c  c | c  c  c  | c  c  c }
    Algo & $t^{gm}_{min}$ & $t^{gm}_{mean}$ & $t^{gm}_{max}$ & $QmC^{gm}_{min}$ &
    $QmC^{gm}_{mean}$ & $QmC^{gm}_{max}$ & $QmD^{gm}_{min}$ &
    $QmD^{gm}_{mean}$ & $QmD^{gm}_{max}$ & $QaD^{gm}_{min}$ &
    $QaD^{gm}_{mean}$ & $QaD^{gm}_{max}$\\ \hline
    $\random$     & 0.094 & 0.112 & 0.145 & 1.849 & 1.836 & 1.817 & 0.820 & 0.812 & 0.858 & 5.086 & 5.603 & 5.046\\
    $\rcm$        & 0.412 & 0.504 & 0.641 & 1.510 & 1.575 & 1.623 & 0.824 & 0.871 & 0.951 & 3.840 & 4.773 & 4.528\\
    $\durebi$     & 415.0 & 445.1 & 484.8 & 0.881 & 0.911 & 0.948 & {\bf 0.617} & 0.636 & 0.685 & 0.951 & 1.207 & 1.477\\
    $\greedyall$  & 72.81 & 75.82 & 79.82 & 1.152 & 1.150 & 1.143 & 0.914 & 0.875 & 0.872 & 1.992 & 2.448 & 2.514\\
    $\greedymin$  & 3.307 & 4.141 & 4.316 & 1.047 & 1.050 & 1.046 & 0.867 & 0.849 & 0.858 & {\bf 0.681} & {\bf 0.584} & {\bf 0.437}\\
    $\greedyallc$ & 96.07 & 99.59 & 103.2 & {\bf 0.720} & {\bf 0.728} & {\bf 0.730} & 0.650 & {\bf 0.631} & {\bf 0.642} & 0.784 & 0.804 & 0.830\\
    $\greedyminc$ & 85.73 & 88.87 & 92.26 & 0.802 & 0.843 & 0.873 & 0.642 & 0.668 & 0.715 & 2.900 & 3.304 & 2.860\\ \hline
  \end{tabular}
}
\end{center}
\label{tab:social:2DGrid2}
\end{table*}

\begin{table*}[!h]
\caption{Mapping of complex networks onto 3DGrid($8 \times 8 \times 8$). Times
  $t^{gm}_{min}$, $t^{gm}_{mean}$ and $t^{gm}_{max }$ are in {\bf milliseconds}.}
\begin{center}
\scalebox{0.9}{
  \begin{tabular}{l || c c c | c  c  c | c  c  c  | c  c  c }
    Algo & $t^{gm}_{min}$ & $t^{gm}_{mean}$ & $t^{gm}_{max}$ & $QmC^{gm}_{min}$ &
    $QmC^{gm}_{mean}$ & $QmC^{gm}_{max}$ & $QmD^{gm}_{min}$ &
    $QmD^{gm}_{mean}$ & $QmD^{gm}_{max}$ & $QaD^{gm}_{min}$ &
    $QaD^{gm}_{mean}$ & $QaD^{gm}_{max}$\\ \hline
    $\random$     & 0.049 & 0.057 & 0.074 & 1.539 & 1.527 & 1.519 & 0.878 & 0.850 & 0.884 & 3.459 & 3.872 & 3.479\\
    $\rcm$        & 0.209 & 0.258 & 0.335 & 1.381 & 1.410 & 1.437 & 0.889 & 0.880 & 0.900 & 2.686 & 3.307 & 3.177\\
    $\durebi$     & 254.0 & 275.0 & 230.0 & 0.935 & 0.952 & 0.967 & {\bf 0.751} & {\bf 0.757} & {\bf 0.751} & {\bf 0.911} & 1.197 & 1.504\\
    $\greedyall$  & 20,78 & 21.49 & 22.47 & 1.108 & 1.104 & 1.112 & 1.134 & 1.083 & 1.035 & 1.412 & 1.538 & 1.548\\
    $\greedymin$  & 0.882 & 0.905 & 0.927 & 1.097 & 1.100 & 1.100 & 0.930 & 0.883 & 0.879 & 0.937 & {\bf 0.856} & {\bf 0.658}\\
    $\greedyallc$ & 25.42 & 26.63 & 28.10 & {\bf 0.791} & {\bf 0.793} & {\bf 0.802} & 0.849 & 0.804 & 0.794 & 0.988 & 1.008 & 0.978\\
    $\greedyminc$ & 23.38 & 24.43 & 25.65 & 0.857 & 0.886 & 0.907 & 0.811 & 0.791 & 0.804 & 3.522 & 4.509 & 4.087\\ \hline
  \end{tabular}
}
\end{center}
\label{tab:app:social:3DGrid}
\end{table*}

\begin{table*}[!h]
\caption{Mapping of complex networks onto 2DTorus($32 \times 32$). Times
  $t^{gm}_{min}$, $t^{gm}_{mean}$ and $t^{gm}_{max }$ are in {\bf milliseconds}.}
\begin{center}
\scalebox{0.9}{
  \begin{tabular}{l || c c c | c  c  c | c  c  c  | c  c  c }
    Algo & $t^{gm}_{min}$ & $t^{gm}_{mean}$ & $t^{gm}_{max}$ & $QmC^{gm}_{min}$ &
    $QmC^{gm}_{mean}$ & $QmC^{gm}_{max}$ & $QmD^{gm}_{min}$ &
    $QmD^{gm}_{mean}$ & $QmD^{gm}_{max}$ & $QaD^{gm}_{min}$ &
    $QaD^{gm}_{mean}$ & $QaD^{gm}_{max}$\\ \hline
    $\random$     & 0.093 & 0.110 & 0.138 & 1.748 & 1.733 & 1.720 & {\bf 0.627} & {\bf 0.644} & {\bf 0.691} & 4.906 & 5.019 & 4.675\\
    $\rcm$        & 0.420 & 0.513 & 0.650 & 1.524 & 1.577 & 1.618 & 0.723 & 0.762 & 0.824 & 3.820 & 4.422 & 4.475\\
    $\durebi$     & 412.5 & 442.1 & 480.1 & 0.925 & 0.946 & 0.958 & 0.697 & 0.764 & 0.826 & 0.942 & 1.318 & 1.971\\
    $\greedyall$  & 74.24 & 75.93 & 78.76 & 1.171 & 1.181 & 1.200 & 0.989 & 0.973 & 0.993 & 1.750 & 2.176 & 2.493\\
    $\greedymin$  & 4.428 & 4.475 & 4.532 & 1.115 & 1.116 & 1.107 & 0.941 & 0.958 & 0.993 & 0.785 & 0.710 & 0.557\\
    $\greedyallc$ & 96.29 & 99.57 & 103.6 & {\bf 0.770} & {\bf 0.769} & {\bf 0.764} & 0.749 & 0.755 & 0.771 & {\bf 0.578} & {\bf 0.522} & {\bf 0.439}\\
    $\greedyminc$ & 86.33 & 89.39 & 93.14 & 0.837 & 0.838 & 0.839 & 0.735 & 0.760 & 0.797 & 2.732 & 3.256 & 3.034\\ \hline
  \end{tabular}
}
\end{center}
\label{tab:social:2DTorus2}
\end{table*}

\begin{table*}[!h]
\caption{Mapping of complex networks onto 3DTorus($8 \times 8 \times 8$). Times
  $t^{gm}_{min}$, $t^{gm}_{mean}$ and $t^{gm}_{max }$ are in {\bf
    milliseconds}.}
\begin{center}
\scalebox{0.9}{
  \begin{tabular}{l || c c c | c  c  c | c  c  c  | c  c  c }
    Algo & $t^{gm}_{min}$ & $t^{gm}_{mean}$ & $t^{gm}_{max}$ & $QmC^{gm}_{min}$ &
    $QmC^{gm}_{mean}$ & $QmC^{gm}_{max}$ & $QmD^{gm}_{min}$ &
    $QmD^{gm}_{mean}$ & $QmD^{gm}_{max}$ & $QaD^{gm}_{min}$ &
    $QaD^{gm}_{mean}$ & $QaD^{gm}_{max}$\\ \hline
    $\random$     & 0.050 & 0.058 & 0.072 & 1.479 & 1.474 & 1.469 & {\bf 0.762} & {\bf 0.753} & {\bf 0.773} & 3.029 & 3.498 & 3.251\\
    $\rcm$        & 0.211 & 259.9 & 342.2 & 1.383 & 1.404 & 1.419 & 0.813 & 0.831 & 0.910 & 2.558 & 3.255 & 3.130\\
    $\durebi$     & 257.3 & 279.1 & 302.5 & 0.995 & 1.014 & 1.028 & 0.893 & 0.925 & 0.994 & 1.003 & 1.418 & 1.804\\
    $\greedyall$  & 20.72 & 21.36 & 22.59 & 1.114 & 1.118 & 1.127 & 1.187 & 1.139 & 1.089 & 1.184 & 1.380 & 1.477\\
    $\greedymin$  & 0.841 & 0.866 & 0.886 & 1.100 & 1.099 & 1.095 & 0.956 & 0.945 & 0.945 & 0.893 & 0.816 & 0.634\\
    $\greedyallc$ & 25.48 & 26.71 & 28.15 & {\bf 0.852} & {\bf 0.847} & {\bf 0.839} & 0.974 & 0.958 & 0.943 & {\bf 0.660} & {\bf 0.620} & {\bf 0.519}\\
    $\greedyminc$ & 23.57 & 24.56 & 25.74 & 0.875 & 0.873 & 0.870 & 0.971 & 0.968 & 0.988 & 1.215 & 1.618 & 1.582\\ \hline
  \end{tabular}
}
\end{center}
\label{tab:app:social:3DTorus}
\end{table*}

\begin{table*}[!h]
\caption{Mapping of complex networks onto 2DTorus($16 \times 16$).}
\begin{center}
\scalebox{0.85}{
  \begin{tabular}{l || c c c | c  c  c | c  c  c  | c  c  c }
    Algo & $\mathcal{Q}t^{gm}_{min}$ & $\mathcal{Q}t^{gm}_{mean}$ & $\mathcal{Q}t^{gm}_{max}$ & $\mathcal{Q}mC^{gm}_{min}$ &
    $\mathcal{Q}mC^{gm}_{mean}$ & $\mathcal{Q}mC^{gm}_{max}$ & $\mathcal{Q}mD^{gm}_{min}$ &
    $\mathcal{Q}mD^{gm}_{mean}$ & $\mathcal{Q}mD^{gm}_{max}$ & $\mathcal{Q}aD^{gm}_{min}$ &
    $\mathcal{Q}aD^{gm}_{mean}$ & $\mathcal{Q}aD^{gm}_{max}$\\ \hline
    $\initialM$     & 0.9714 & 1.0045 & 1.0365 & 1.0369 & 1.0357 & 1.0463 & 0.9555 & 0.9859 & 1.0147 & 1.2189 & 1.2921 & 1.3684\\
    $\initialMO$    & 1.0000 & 0.9888 & 1.0098 & 8.7294 & 8.9164 & 9.0664 & 3.2390 & 3.0742 & 2.9958 & 3.2217 & 2.9771 & 2.5263\\
    $\randomM$      & 1.0197 & 1.0232 & 1.0457 & 0.9888 & 0.9887 & 0.9873 & 1.0984 & 1.1191 & 1.1202 & 0.9969 & 1.0419 & 1.1427\\
    $\randomMO$     & 1.0153 & 1.0125 & 1.1086 & 7.6895 & 7.7816 & 7.8792 & 2.8851 & 2.9030 & 2.8530 & 1.8989 & 1.7575 & 1.5831\\
    $\rcmM$         & 1.1057 & 1.0787 & 0.9371 & 0.9843 & 0.9860 & 0.9910 & 1.1042 & 1.1155 & 1.1473 & 0.9775 & 1.0073 & 1.1084\\
    $\rcmMO$        & 0.9755 & 0.9936 & 0.8855 & 7.2693 & 7.3256 & 7.5483 & 3.1131 & 3.3055 & 3.2221 & 1.5371 & 1.5224 & 1.6418\\
    $\durebiM$      & 1.0430 & 1.0454 & 1.0633 & 1.0021 & 1.0092 & 1.0118 & 1.0575 & 1.0732 & 1.1035 & 0.9560 & 1.0777 & 1.1059\\
    $\durebiMO$     & 0.4275 & 0.4097 & 0.3889 & 8.0650 & 8.1325 & 8.4128 & 3.2878 & 3.5476 & 3.5615 & 2.1035 & 2.0679 & 1.6679\\
    $\greedyminM$   & 1.0334 & 1.0333 & 1.0187 & 1.0192 & 1.0181 & 1.0180 & 1.0293 & 1.0652 & 1.1093 & 1.1002 & 1.1122 & 1.1105\\
    $\greedyminMO$  & 0.6853 & 0.6817 & 0.6872 & 8.3516 & 8.4982 & 8.7077 & 3.4357 & 3.4170 & 3.3840 & 2.5083 & 2.4816 & 2.3758\\
    $\greedyallcM$  & 1.0608 & 1.0673 & 1.0932 & 1.0340 & 1.0304 & 1.0287 & 0.9963 & 1.0285 & 1.0338 & 1.1039 & 1.1171 & 1.1198\\
    $\greedyallcMO$ & 0.3670 & 0.3733 & 0.3912 & 8.1861 & 8.3093 & 8.4513 & 4.1153 & 4.1375 & 4.1129 & 2.4514 & 2.3375 & 2.3238\\
    $\greedymincM$  & 1.0827 & 1.0731 & 0.9870 & 1.0304 & 1.0283 & 1.0335 & 1.0148 & 1.0280 & 1.0331 & 1.0229 & 1.0711 & 1.1837\\
    $\greedymincMO$ & 0.3723 & 0.3708 & 0.3484 & 8.1740 & 8.3004 & 8.4825 & 3.7402 & 3.7732 & 3.7904 & 2.0304 & 1.7733 & 1.6501\\ \hline
  \end{tabular}
}
\end{center}
\label{tab:app:socialComp:2DTorus1}
\end{table*}

\end{document}